\newif\ifprivate
\privatetrue

\newif\ifshort

\documentclass[twocolumn,pagebackref]{article}

\usepackage[
	includefoot,
	top=1.6cm,
	bottom=1.3cm,
	left=1.35cm,
	right=1.35cm,
	a4paper,
]{geometry}

\makeatletter
\renewenvironment{abstract}{%
	\small
	\quotation
	\noindent{\bfseries\abstractname.}%
}{\endquotation}
\renewcommand\section{\@startsection {section}{1}{\z@}%
                                   {-2.5ex \@plus -1ex \@minus -.2ex}%
                                   {1.3ex \@plus.2ex}%
                                   {\normalfont\large\bfseries}}
\renewcommand\subsection{\@startsection{subsection}{2}{\z@}%
                                     {-2.25ex\@plus -1ex \@minus -.2ex}%
                                     {0.5ex \@plus .2ex}%
                                     {\normalfont\normalsize\bfseries}}
\renewcommand\subsubsection{\@startsection{subsubsection}{3}{\z@}%
                                     {-2.25ex\@plus -1ex \@minus -.2ex}%
                                     {0.5ex \@plus .2ex}%
                                     {\normalfont\normalsize\bfseries}}
\renewcommand\paragraph{\@startsection{paragraph}{4}{\z@}%
                                    {2.25ex \@plus1ex \@minus.2ex}%
                                    {-1em}%
                                    {\normalfont\normalsize\bfseries}}
\renewcommand\subparagraph{\@startsection{subparagraph}{5}{\parindent}%
                                       {2.25ex \@plus1ex \@minus .2ex}%
                                       {-1em}%
                                      {\normalfont\normalsize\bfseries}}
\makeatother

\usepackage{soul}
\usepackage{url}
\usepackage[x11names, table]{xcolor}
\usepackage[utf8]{inputenc}
\usepackage[small]{caption}
\usepackage{graphicx}
\usepackage{amsmath}
\usepackage{booktabs}

\usepackage[T1]{fontenc}
\usepackage{amssymb}
\usepackage[outline]{contour}
\usepackage{enumerate}
\usepackage{dsfont}
\usepackage{tabularx}
\usepackage{multicol}
\usepackage[ruled,vlined,linesnumbered]{algorithm2e}

\usepackage[square,numbers]{natbib}
\usepackage[linkcolor=red!50!black,citecolor=green!40!black,colorlinks]{hyperref}
\usepackage[capitalize,nameinlink]{cleveref}
\usepackage{doi}
\usepackage{paralist}
\usepackage{tikz}
\usetikzlibrary{calc,positioning}
\usepackage{pgfplots}
\pgfplotsset{major grid style={very thin,gray!20!white}} %
\usepackage{pgfplotstable}
\usepackage{subcaption}

\usepackage{mathtools}

\usepackage{multirow}

\usepackage{pifont}%
\newcommand{\xmark}{\ding{55}}%

\newcommand{\appsymb}{$\bigstar$}
\newcommand{\appref}[1]{\ifshort{}{\hyperref[proof:#1]{\appsymb}}\fi{}}

\usepackage{etoolbox}

\newcommand{\appendixsection}[1]{%
	\ifshort{}\gappto{\appendixProofText}{\section{Additional Material for Section~\ref{#1}}\label{app:#1}}\fi{}
}

\newcommand{\toappendix}[1]{%
	\ifshort{}\gappto{\appendixProofText}
	{{
		#1
	}}\else{}#1\fi{}
}
\newcommand{\appendixproof}[2]{%
	\ifshort{}\gappto{\appendixProofText}
		{
			\subsection{Proof of \cref{#1}}\label{proof:#1}
			#2
		}\else{}#2\fi{}
}

\usepackage{xargs} 
\newcommandx{\set}[2][1=1]{\ensuremath{\{#1,\ldots,#2\}}}
\newcommandx{\tlog}[3][1=,3=]{\log_{#1}^{#3}(#2)}

\usepackage{amsthm}
\usepackage{thmtools}
\theoremstyle{plain}
\newtheorem{theorem}{Theorem}
\newtheorem{lemma}{Lemma}
\newtheorem{corollary}{Corollary}
\newtheorem{observation}{Observation}
\newtheorem{proposition}{Proposition}
\newtheorem{rrule}{Reduction Rule}
\theoremstyle{definition}
\newtheorem{definition}{Definition}
\newtheorem{problem}{Problem}
\declaretheorem[style=definition,name=Construction,qed=$\diamond$]{construction}
\theoremstyle{remark}
\newtheorem*{remark}{Remark}
 
\crefname{observation}{Observation}{Observations}
\crefname{rrule}{Reduction Rule}{Reduction Rules}
\crefname{construction}{Construction}{Constructions}
\crefname{proposition}{Proposition}{Propositions}
\crefname{theorem}{Theorem}{Theorems}
\crefname{corollary}{Corollary}{Corollaries}
\crefname{line}{Line}{Lines}
\crefname{problem}{Problem}{Problems}
\crefname{figure}{Figure}{Figures}
\crefname{subsection}{Section}{Sections}
\crefname{section}{Section}{Sections}

\Crefname{subsection}{Sec.}{Sects.}
\Crefname{section}{Sec.}{Sects.}
\Crefname{problem}{Prob.}{Probs.}
\Crefname{observation}{Obs.}{Obs.}
\Crefname{proposition}{Prop.}{Props.}
\Crefname{corollary}{Cor.}{Cors.}
\Crefname{theorem}{Thm.}{Thm.}

\newcommandx{\decprob}[6][3=Input,5=Question]{\begin{samepage}
  \begingroup
\begin{problem}\label{prob:#2}%
  {{\textsc{#1}}}
  \nopagebreak[4]\end{problem}\nopagebreak[4]\vspace{-0.5em}
  \par\noindent\hangindent=\parindent\textbf{#3}:  #4\nopagebreak[4]
  \par\noindent\hangindent=\parindent\textbf{#5}:  #6
  \par\medskip
  \endgroup
  \end{samepage}
}

\newcommand{\yes}{\textnormal{\texttt{yes}}}

\newcommand{\RD}{$(\Rightarrow)\quad$}
\newcommand{\LD}{$(\Leftarrow)\quad$}

\newcommandx{\decprobX}[5][2=Input,4=Question]{%
  \begingroup
  \par\medskip
  \noindent \colorbox{gray!17!white}{\textsc{#1}\index{problem!#1}}\nopagebreak[4]
  \par\noindent\hangindent=\parindent\textbf{#2}:  #3\nopagebreak[4]
  \par\noindent\hangindent=\parindent\textbf{#4}:  #5
  \par  \medskip
  \endgroup
}

\newcommand{\N}{\mathbb{N}}
\newcommand{\Nzero}{\mathbb{N}_0}

\newcommand{\calC}{\mathcal{C}}

\newcommand{\prob}[1]{\textnormal{\textsc{#1}}}

\newcommand{\gbp}{Green Bridges Placement}

\newcommand{\GBP}{\prob{GBP}}

\newcommandx{\RgbpTsc}[2][1=$1$]{\prob{#1-Reach \gbp}}
\newcommandx{\RgbpAcr}[2][1=$1$]{\prob{#1-Reach \GBP}}
\newcommandx{\RwgbpTsc}[2][1=$1$]{\prob{#1-Reach \gbp{} with Costs}}
\newcommandx{\RwgbpAcr}[2][1=$1$]{\prob{#1-Reach \GBP-C}}
\newcommandx{\CgbpTsc}[2][1=$d$]{\prob{#1-Closed \gbp}}
\newcommandx{\CgbpAcr}[2][1=$d$]{\prob{#1-Closed \GBP}}
\newcommandx{\DgbpTsc}[2][1=$d$]{\prob{#1-Diamater \gbp}}
\newcommandx{\DgbpAcr}[2][1=$d$]{\prob{#1-Diam \GBP}}

\newcommand{\mcmTsc}{\prob{Maximum-Weight Matching}}
\newcommand{\mcmAcr}{\prob{MWM}}
\newcommand{\mchmTsc}{\prob{Maximum-Weight Hypergraph Matching}}
\newcommand{\mchmAcr}{\prob{MWHM}}
\newcommand{\cvcTsc}{\prob{Cubic Vertex Cover}}
\newcommand{\cvcAcr}{\prob{CVC}}

\newcommand{\vcAcr}{\prob{VC}}

\newcommand{\cocl}[1]{\ensuremath{\operatorname{#1}}}

\newcommand{\classP}{\cocl{P}}
\newcommand{\NP}{\cocl{NP}}

\newcommand{\calA}{\mathcal{A}}

\newcommand{\calH}{\mathcal{H}}

\newcommandx{\tref}[2][1=]{{\scriptsize{(\Cref{#2}#1)}}}

\newcommand{\ceq}{\ensuremath{\coloneqq}}

\newcommand{\crown}{crown}
\newcommand{\pqcrown}[2]{$(#1,#2)$-\crown}

\definecolor{lilla}{HTML}{750787}

\newcommand{\thecolor}{black}%

\usepackage{makecell}

\ifprivate{}
	\usepackage[text size=footnotesize,color=green!15!white]{todonotes}
\else{}
	\usepackage[text size=footnotesize,color=green!15!white,disable]{todonotes}
\fi{}

\usepackage{etoolbox}
\newcommand{\ExternalLink}{%
\tikz[x=1.2ex, y=1.2ex, baseline=-0.05ex]{%
    \begin{scope}[x=1ex, y=1ex]
        \clip (-0.1,-0.1) --++ (-0, 1.2) --++ (0.6, 0) --++ (0, -0.6) --++ (0.6, 0) --++ (0, -1);
        \path[draw, line width = 0.5, rounded corners=0.5] (0,0) rectangle (1,1);
    \end{scope}
    \path[draw, line width = 0.5] (0.5, 0.5) -- (1, 1);
    \path[draw, line width = 0.5] (0.6, 1) -- (1, 1) -- (1, 0.6);
    }
}

\newcommand{\tikzpreamble}{%
  \tikzstyle{xnode}=[circle,fill,scale=0.5,draw,color=\thecolor]
  \tikzstyle{xedge}=[thick,-,color=\thecolor]
  \tikzstyle{xxedge}=[ultra thick,-,color=\thecolor]
  \tikzstyle{xedgedot}=[thick,-,dotted,color=white]
  \tikzstyle{xhabA}=[-,opacity=0.125, line width=9pt, line cap=round,color=magenta]
  \tikzstyle{xhabB}=[-,opacity=0.125, line width=9pt, line cap=round,color=green]
  \tikzstyle{xhabC}=[-,opacity=0.125, line width=9pt, line cap=round,color=cyan]
  \tikzstyle{xhabD}=[-,opacity=0.125, line width=9pt, line cap=round,color=blue]
  \tikzstyle{xhabE}=[-,opacity=0.125, line width=9pt, line cap=round,color=yellow]
}

\usepackage{authblk}

\newcommand{\mytitle}{Placing Green Bridges Optimally,
with Habitats Inducing Cycles}

\title{\Large\bf \mytitle}
\author{Maike Herkenrath}
\author{Till Fluschnik\footnote{Supported by DFG, project MATE (NI/369-17).}}
\author{Francesco Grothe\footnote{Supported by DFG, project MATE (NI/369-17).}}
\author{Leon Kellerhals}

\affil{\small
  Technische Universit\"at Berlin, Faculty~IV, Institute of Software Engineering and Theoretical Computer Science, Algorithmics and Computational Complexity.\protect\\
  \texttt{\{herkenrath,\,f.grothe\}@campus.tu-berlin.de,\,\{till.fluschnik,\,leon.kellerhals\}@tu-berlin.de}
}

\date{}

\begin{document}

\maketitle

\begin{abstract}
Choosing the placement of wildlife crossings (i.e., green bridges)
to reconnect animal species' fragmented habitats is among the 17 goals towards sustainable development by the UN.
We consider the following established model:
Given a graph whose vertices represent the fragmented habitat areas and whose weighted edges represent possible green bridge locations,
as well as the habitable vertex set for each species, find the cheapest set of edges such that each species' habitat is connected.
We study this problem from a theoretical (algorithms and complexity) and an experimental perspective,
while focusing on the case where habitats induce cycles.
We prove that the NP-hardness persists in this case even if the graph structure is restricted. %
If the habitats additionally induce faces in plane graphs however, the problem becomes efficiently solvable. %
In our empirical evaluation we compare this algorithm as well as ILP formulations for more general variants and an approximation algorithm with another.
Our evaluation underlines that each specialization is beneficial in terms of running time, whereas the approximation provides highly competitive solutions in practice.
\end{abstract}

\section{Introduction}
\label{sec:intro}

Habitat fragmentation due to human-made infrastructures 
like roads or train tracks,
leading to wildlife-vehicle collisions,
a severe threat not only to animals,
up to impacting biodiversity~\cite{bennett2017effects,SawayaKC14},
but also to humans.
Installing wildlife crossings like bridges, tunnels~\cite{WoltzGD08}, ropes~\cite{goldingay2017can},
et cetera 
(we refer to those as \emph{green bridges} from here on)
in combination with road fencing (so as to ensure that the green bridges are being used)
allows a cost-efficient~\cite{HuijserDCAM09}
reduction of wildlife-vehicle collisions by up to 85\%~\cite{HuijserMHKCSA08}.
In this paper,
we study the problem of finding the right positions for green bridges that keeps the building cost at a minimum and ensures that every habitat is fully interconnected.
We focus on those cases in which the structure of the habitats is very simple
and study the problem from both a theoretical (algorithmics and computational complexity) as well as an experimental perspective.

We follow a model recently introduced by~\citet{FK21}.
Herein,
the modeled graph can be understood as path-based graph~\cite{UrbanMTS09,galpern2011patch}:
a vertex corresponds to a part fragmented by human-made infrastructures subsuming habitat patches
of diverse animal habitats,
and any two vertices are connected by an edge if the corresponding patches can be connected by a green bridge.
The edges are equipped with the costs of building the respective green bridge (possibly including fencing) in the respective area.
The goal is to construct green bridges in a minimum-cost way such that
in the graph spanned by the green bridges,
the patches of each habitat form a connected component.
Formally, we are concerned with the following.

\decprob{\RwgbpTsc[1]{} (\RwgbpAcr[1]{})}{rwgbp}
{An undirected graph~$G=(V,E)$ with edge costs $c\colon E\to \N$,
a set~$\calH=\{H_1,\dots,H_r\}$ of habitats with $H_i\subseteq V$ and~$|H_i|\geq 2$ for all~$i\in\set{r}$, 
an integer~$k\in\Nzero$.}
{Is there a subset~$F\subseteq E$ with~$c(F)\ceq \sum_{e\in F}c(e)\leq k$ such that
for every~$i\in\set{r}$
it holds that $H_i\subseteq V(G[F])$ and~$G[F][H_i]$ is a connected graph?
}

\noindent
In accordance with~\citet{FK21},
we denote by
\RgbpTsc[1]{} (\RgbpAcr[1]{})
the unit-cost version of~\RwgbpAcr[1]{}.

\paragraph{Our contributions.}
Our study focuses on habitats which induce small cycles.
This is well motivated from practical as well as theoretical standpoints.
Small habitats, in terms of size and limited structures (as to trees and cycles),
appear more often for small mammals, 
amphibians,
and reptiles,
among which several species are at critical state~\cite{HammerM08,hennings2010wildlife}.
From a theoretical point of view,
since the problem is already \NP-hard in quite restricted setups~\cite{FK21},
it is canonical to study special cases such as restrictions on habitat and graph structure.
As the problem is polynomial-time solvable if each habitat induces a tree (\cref{obs:habtree}),
studying habitats inducing cycles are an obvious next step.

\begin{table}[t]
 \centering
 \caption{Our \NP-hardness (refer to~\cref{thm:lbs}) and upper bound results regarding our habitat families. 
 $\calC=\bigcup_{\ell\in\N_{\geq 3}} \{C_\ell\}$ denotes the class of all cycles~$C_\ell$ of length~$\ell$.
 ``P'' is short for ``polynomial-time solvable'',
 $\Delta$ for~$\Delta(G)$.
 $^*$\,($\ell\neq 5$)
 $^\dagger$\,(if every edge is in at most two habitats~\tref{thm:habindface})
 $^\ddagger$\,(if~$\Delta\leq 2$~\tref{obs:degtwo})
 $^\S$\,(if~$\Delta\leq 3$ \tref{cor:habindface})}
 \begin{tabular}{@{}p{2.7cm}p{4.1cm}l@{}}\toprule
 Habitat family & \NP-hard, even if & Upper \\\midrule\midrule
  $\{P_2,C_3\}$ & $G$ is a clique & \classP$^{\dagger\lor\ddagger}$  \\
  $\{C_3\}$ & \emph{(no further restrictions)} & \classP{}$^{\dagger\lor \S}$  \\\midrule
  $\{P_2,C_\ell\}$, ${\ell\in\N_{\geq 4}}$ & $\Delta\geq 4$ or if~$G$ is planar$^*$ & \classP$^{\dagger\lor\ddagger}$ \\ %
  $\{C_\ell\}$, ${\ell\in\N_{\geq 4}}$ & $\Delta\geq 10$ or if~$G$ is planar$^*$  & \classP$^{\dagger\lor\ddagger}$ \\\midrule
  $\{P_2\}\cup\calC$ & $\Delta\geq 3$ and~$G$ is planar & \classP$^{\dagger\lor\ddagger}$  \\
  $\calC$ & $\Delta\geq 9$ and $G$ is planar  & \classP$^{\dagger\lor\ddagger}$ \\ 
  \bottomrule
 \end{tabular}
 \label{tab:results}
\end{table}
Our theoretical results are summarized in~\cref{tab:results}.
We prove that \RgbpAcr{} remains \NP-hard even if each habitat induces a cycle,
even of fixed length at least three.
Most of our \NP-hardness results hold even on restricted input graphs,
i.e.,
planar graphs of small maximum degree.
On the positive side, we prove that for cycle-inducing habitats
we can reduce the problem to maximum-weight matching in an (auxiliary) multi-hypergraph.
From this we derive a polynomial-time special case:
If every edge is shared by at most two habitats,
we can reduce the problem to maximum-weight matching,
a well-known polynomial-time solvable problem.

We perform an experimental evaluation of several algorithms,
including the two mentioned above,
the approximation algorithm given by~\citet{FK21},
as well as an ILP formulation for the case of general
(i.e., not necessarily inducing cycles) 
habitat structures.
Our evaluation shows that each more specialized algorithm 
for the cycle-inducing habitats
perform much better than the next more general one.
Moreover,
we show that the approximation algorithm is fast
with very small loss in solution quality.
Finally,
we underline the connection between solution quality and running time on the one side,
as well as the way of how habitats intersect on the other side.

\paragraph{Further related work.}
\citet{Ament15} gives an informative overview on the topic.
The placement of wildlife crossings
is also studied
with different approaches~\cite{clevenger2002gis,DownsHLAKO14,LoraammD16,BastilleWDW18}.
A related problem is \prob{Steiner Forest} 
(where we only need to connect habitats, 
without requesting connected induced graphs),
which (and an extension of it) 
is studied from an algorithmic perspective~\cite{LaiGSMCM11,JordanS15}.

\section{Preliminaries}
\label{sec:prelims}
\ifshort{}\appendixsection{sec:prelims}\fi{}

Let~$\N$ ($\Nzero$) be the natural numbers without (with) zero.
For a set~$X\subseteq \Nzero$ and~$y\in\Nzero$,
let~$X_{\geq y}\ceq \{x\in X\mid x\geq y\}$.

\paragraph*{Graph Theory.}
For a graph~$G=(V,E)$,
we also denote by~$V(G)$ and~$E(G)$ the vertex set~$V$ and edge set~$E$ of~$G$,
respectively.
For an edge set~$E'\subseteq E$ and vertex set~$V'\subseteq V$,
we denote by~$G[E']=(\bigcup_{e\in E'}e,E')$
and by~$G[V']\ceq (V',\{e\in E\mid e\subseteq V'\})$ the graph induced by~$E'$ and by~$V'$,
respectively.
The graph~$G[E'][V']$ is the graph~$G'[V']$ with~$G'=G[E']$.
By~$\Delta(G)$ ($\delta(G)$) we denote the maximum (minimum) vertex degree of~$G$.
Denote by~$N_G(v)\ceq \{w\in V(G)\mid \{v,w\}\in E(G)\}$ and~$N_G[v]\ceq N_G(v)\cup\{v\}$
the open and closed neighborhood of~$v$ in~$G$.

\paragraph*{Basic Observations.}

We can assume that every vertex in our graph is contained in a habitat.
We say that an edge~$e$ is \emph{shared} by two habitats~$H,H'$ if~$e\subseteq H\cap H'$.
We say that a set~$F' \subseteq E$ \emph{satisfies} a habitat~$H \in \calH$ if $H\subseteq V(G[F'])$ and~$G[F'][H]$ is connected.
We have the following.

\begin{observation}%
 \label{obs:habtree}
 \RwgbpAcr{} where each habitat induces a tree
 is solvable in~$O(|\calH|\cdot |G|)$ time.
\end{observation}

\appendixproof{obs:habtree}
{
\begin{proof}
 Since each habitat~$H\in\calH$ induces a tree,
 we need to take all edges of~$G[H]$ into the solution.
 Hence~$F=\bigcup_{H\in \calH} E(G[H])$ is a minimum-cost solution computable in~$O(|\calH|\cdot |G|)$ time.
\end{proof}
}

\begin{observation}%
 \label{obs:degtwo}
 \RwgbpAcr{} on graphs of maximum degree two is solvable in~$O(|\calH|\cdot |G|)$ time.
\end{observation}

\appendixproof{obs:degtwo}
{
\begin{proof}
 Every connected component is a cycle or a path.
 Hence,
 every habitat is either a cycle or a path.
 In a connected component which is a cycle,
 all edges induced by habitats inducing paths must be taken.
 If not all edges are taken,
 then we can leave out exactly one remaining edge of largest cost from the solution.
\end{proof}
}

\section{Lower Bounds}
\label{sec:lbs}
\ifshort{}\appendixsection{sec:lbs}\fi{}

In this section we show that the \NP-hardness of \RgbpAcr{} persists even if the habitats induce simple structures and the graphs are restricted.
We prove the following.
\begin{theorem}
 \label{thm:lbs}
 \RgbpAcr{} is \NP-hard even if:
 \begin{compactenum}[(i)]
 \item each habitat induces a $P_2$ or a~$C_3$ and~$G$ is a clique, 
  or each habitat induces a~$C_\ell$ for any fixed~$\ell\geq 3$.
  \item each habitat induces (a $P_2$ or) a~$C_\ell$ for any fixed $\ell\in\N_{\geq4}\setminus\{5\}$ and~$G$ is planar.
  \item each habitat induces a~$P_2$ or a~$C_\ell$ for any fixed $\ell\geq 4$ and $\Delta(G)\geq 4$,
    or each habitat induces a~$C_\ell$ for any fixed $\ell\geq 4$ and $\Delta(G)\geq 10$.
  \item each habitat induces a~$P_2$ or a cycle, $G$ is planar, and~$\Delta(G)\geq 3$, 
    or each habitat induces a cycle, $G$ is planar, and~$\Delta(G)\geq 9$.
 \end{compactenum}
\end{theorem}

\noindent
For each case we provide a polynomial-time reduction from the following \NP-hard \cite{GareyJ79} problem.

\decprob{\cvcTsc{} (\cvcAcr{})}{cvc}
{An undirected, cubic graph~$G=(V,E)$ and an integer~$p\in\Nzero$.}
{Is there a set~$V'\subseteq V$ with~$|V'|\leq p$ such that~$G[V\setminus V']$ contains no edge?}

\noindent
We first provide constructions for the four cases
before presenting 
the correctness proofs.

\paragraph{Constructions.}
We next provide the constructions for the base cases (i.e., small cycle lengths and habitats inducing $P_2$s) of \cref{thm:lbs}(i)--(iv).
The results can be extended by employing a central gadget which we call \emph{crown}.
The crown allows us to exclude $P_2$s from the habitat family and to extend cycle lenghts all while preserving the planarity of the reductions.
See \cref{fig:reductions}(a) for a crown.

\begin{figure*}[t]
 \centering
 \begin{tikzpicture}
  \def\xr{0.8}
  \def\yr{1}
  \tikzpreamble{}
  
  \begin{scope}[shift={(-4*\xr,0*\yr)}]
    \node at (-1.5*\xr,1.25*\yr)[]{(b)};
    \node (t) at (0*\xr,-1*\yr)[xnode]{};
    \foreach \l/\x in {1/-1.5,i/-0.5,j/0.5,n/1.5}{\node (v\l) at (\x*\xr,0.5*\yr)[xnode]{};
    \draw[xedge] (v\l) to (t);}
    \foreach \x in {-1,0,1}{\node at (\x*\xr,0.5*\yr)[]{$\cdots$};}
    \foreach \a/\b/\h in {1/i/A,i/n/B,j/n/C}{
      \draw[xxedge,color=red] (v\a) to [out=45,in=135](v\b);
      \draw[xhab\h] (v\a) to [out=45,in=135](v\b);
      \draw[xhab\h] (v\a) to (t);\draw[xhab\h] (v\b) to (t);
    }
    \foreach \a/\b in {1/j,1/n,i/j}{\draw[-,thick,dashed] (v\a) to [out=45,in=135](v\b);}
    \end{scope}
  
  \begin{scope}
    \node at (-1.5*\xr,1.25*\yr)[]{(c)};
    \node (s) at (0*\xr,1*\yr)[xnode]{};
    \node (t) at (0*\xr,-1*\yr)[xnode]{};
    \foreach \l/\x in {1/-1.5,i/-0.5,j/0.5,n/1.5}{\node (v\l) at (\x*\xr,0*\yr)[xnode]{};\draw[xxedge,color=red] (s) to (v\l);\draw[xedge] (v\l) to (t);}
    \foreach \x in {-1,0,1}{\node at (\x*\xr,0*\yr)[]{$\cdots$};}
    \foreach \a/\b/\h in {1/i/A,i/n/B,j/n/C}{\draw[xhab\h] (s) to (v\a) to (t) to (v\b) to (s);}
  \end{scope}
  
  \begin{scope}[shift={(4*\xr,0*\yr)}]
    \node at (-1.5*\xr,1.25*\yr)[]{(d)};
    \foreach \l/\x in {1/-1.5,i/-0.5,j/0.5,n/1.5}{\node (v\l) at (\x*\xr,0.5*\yr)[xnode]{};
    \node (w\l) at (\x*\xr,-0.5*\yr)[xnode]{};
    \draw[xedge] (v\l) to (w\l);}
    \foreach \x in {-1,0,1}{\node at (\x*\xr,0*\yr)[]{$\cdots$};}
    \foreach \a/\b/\h in {1/i/A,i/n/B,j/n/C}{
      \draw[xxedge,color=red] (v\a) to [out=45,in=135](v\b);
      \draw[xhab\h] (v\a) to [out=45,in=135](v\b);
      \draw[xedge] (w\a) to [out=-45,in=-135](w\b);\draw[xhab\h] (w\a) to [out=-45,in=-135](w\b);
      \draw[xhab\h] (v\a) to (w\a);\draw[xhab\h] (v\b) to (w\b);
    }
    \end{scope}
    
    \begin{scope}[shift={(-7.5*\xr,-1*\yr)}]
    \node at (-1*\xr,1.25*\yr+1*\yr)[]{(a)};
    
    \node (a) at (-1*\xr,0*\yr)[xnode,label=-90:{$a$}]{};
    \node (b) at (1*\xr,0*\yr)[xnode,label=-90:{$b$}]{};
    \node (c1) at (0*\xr,1*\yr*\yr)[xnode]{};
    \node (c2) at (0*\xr,2*\yr*\yr)[xnode]{};
    
    \foreach\x/\y in {a/b,a/c1,c1/b,a/c2,c2/b}{\draw[xedge] (\x) to node[midway,xnode](D\x\y){}(\y);}
    \foreach\x/\y/\z/\h in {a/c1/b/D,a/c2/b/E}{\draw[xhab\h] (\x) to (\y) to (\z) to (\x);}
    \foreach\x/\y in {a/Dab,Dab/b,a/Dac1,Dac1/c1,c1/Dc1b,a/Dac2,Dac2/c2,c2/Dc2b}{\draw[xxedge,color=green!66!black] (\x) to (\y);}
    
  \end{scope}
  
  \begin{scope}[shift={(7.5*\xr,-1*\yr)}]
    \node at (-1*\xr,1.25*\yr+1*\yr)[]{(e)};
  \end{scope}
 \end{tikzpicture}
 \begin{tikzpicture}[every node/.style = {circle,fill,scale=0.33,draw,color=\thecolor},level distance=20pt,level/.style = {sibling distance = 12mm/#1},edge from parent/.style={draw,ultra thick}]
  \def\xr{0.8}
  \def\yr{1}
  \tikzpreamble{}

  \newcommand{\bintree}[1]{
    \node (#1x1){}
    child {node (#1x2){} 
          child {node (#1x4){}
            child {node (#1x8){}}
            child {node (#1x9){}}
                }
        child {node (#1x5){}
               child {node (#1x10){}}
               child {node (#1x11){}}
              }
        }
  child {node (#1x3){} 
        child {node (#1x6){}
          child {node (#1x12){}}
          child {node (#1x13){}}
              }
        child {node (#1x7){}
               child {node (#1x14){}}
               child {node (#1x15){}}
              }
        };
  }
  
  \begin{scope}[rotate=90]
    \bintree{A}
  \end{scope}
  
  \begin{scope}[shift={(6*\xr,0)},xscale=-1,rotate=90]
    \bintree{B}
  \end{scope}
  
  \foreach \x in {8,...,15}{\pgfmathsetmacro\y{int(\x-7)};\draw[xedge] (Ax\x) to node[midway,below,draw=none,fill=none,scale=1.75,yshift=4pt]{$v_{\y}$}(Bx\x);}
  \foreach \a in {8,10}{
      \draw[xhabA] (Ax\a) to (Bx\a);
  }
  \foreach \a/\b in {8/4,4/2,2/5,5/10}{
      \draw[xhabA] (Ax\a) to (Ax\b);
      \draw[xhabA] (Bx\a) to (Bx\b);
  }
  
  \foreach \a in {11,15}{
      \draw[xhabB] (Ax\a) to (Bx\a);
  }
  \foreach \a/\b in {11/5,5/2,2/1,1/3,3/7,7/15}{
      \draw[xhabB] (Ax\a) to (Ax\b);
      \draw[xhabB] (Bx\a) to (Bx\b);
  }
 \end{tikzpicture}
 \caption{(a) A crowning and (b)--(e) four construction types. 
 (a) The graph introduced by a \pqcrown{1}{3}ing with its two \crown{}-habitats (blue and yellow).
 Thick (green) edges form a minimum cardinality solution for the \crown{}-habitats with 8 edges.
 (b)--(d) The construction are exemplified for a graph containing the edges~$\{v_1,v_i\}$, $\{v_i,v_n\}$, and $\{v_j,v_n\}$,
 but not the edges $\{v_1,v_j\}$, $\{v_1,v_n\}$, and $\{v_i,v_j\}$,
 (b)+(c)
 Three habitats corresponding to the edges~$\{v_1,v_i\}$ (magenta),
 $\{v_i,v_n\}$ (green), 
 and~$\{v_j,v_n\}$ (blue) are depicted.
 Endpoints of thick (red) edges are candidates for \pqcrown{p}{q}ing.
 In (b),
 dashed edge can be added without changing the correctness.
 (d) Thick (red) edges are candidates for subdivisions.
 (e)
 Habitats for the edges~$\{v_1,v_3\}$ and~$\{v_4,v_8\}$ are depicted. 
 Each of the thick edges forms a habitat.
 }
 \label{fig:reductions}
\end{figure*}
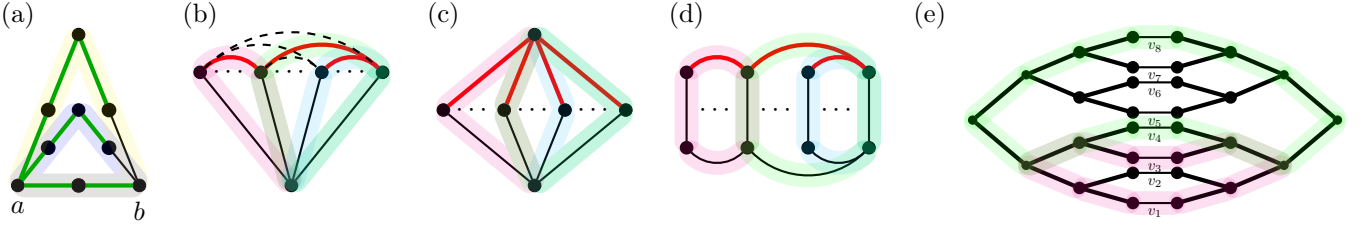

\begin{definition}
 \label{def:crown}
 Let~$G$ be a graph with two distinct vertices~$a,b\in V(G)$ and habitats~$\calH$.
 When we say we \pqcrown{p}{q} $a$ and~$b$,
 then 
 we connect~$a$ and~$b$ 
 by a so-called base path~$P^0$ of length~$p+1$ and 
 two \crown-paths~$P^1,P^2$, each of length~$q+1$,
 and add two \crown-habitats $H_i=V(P^0)\cup V(P^i)$ for~$i\in\{1,2\}$.
\end{definition}

\begin{observation}
 The minimum number of edges to satisfy both crown habitats of a $(p,q)$-crowning
 is~$p+2q+1$ (every edge from the base path and in each \crown-path, all but one edge).
\end{observation}

\begin{construction}
 \label{constr:c3}
 Let~$I=(G,p)$ be an instance of \cvcAcr{} with~$G=(V,E)$, $V=\{v_1,\dots,v_n\}$, and~$E=\{e_1,\dots,e_m\}$.
 Construct an instance~$I'=(G',\calH,k)$ as follows
 (see~\cref{fig:reductions}(b)).
 Let~$G'=(V',E')$ with $V'\ceq V\cup\{x\}$ and $E'\ceq E\cup\bigcup_{i=1}^n \{\{x,v_i\}\}$.
 Let~$\calH=\{H_1,\dots,H_m\}\cup \{H_1',\dots,H_m'\}$,
 where~$H_i\ceq e_i$ and~$H_i'\ceq e_i\cup \{x\}$ for all~$i\in\set{m}$.
 Let~$k\ceq m+p$.
\end{construction}

\begin{construction}
 \label{constr:planar}
 Let~$I=(G,p)$ be an instance of \cvcAcr{} with~$G=(V,E)$, $V=\{v_1,\dots,v_n\}$, and~$E=\{e_1,\dots,e_m\}$.
 Construct an instance~$I'=(G',\calH,k)$ as follows
 (see~\cref{fig:reductions}(c)).
 Let~$G'=(V',E')$ with $V'\ceq V\cup\{x,y\}$ and~$E'\ceq \bigcup_{i=1}^n \{\{x,v_i\},\{y,v_i\}\}$.
 Let~$\calH=\{H_1,\dots,H_n\}\cup \{H_1',\dots,H_m'\}$,
 where~$H_i\ceq \{v_i,x\}$ for all~$i\in\set{n}$ and~$H_i'\ceq e_i\cup \{x,y\}$ for all~$i\in\set{m}$.
 Let~$k\ceq n+p$.
\end{construction}

\begin{construction}
 \label{constr:deg}
 Let~$I=(G,p)$ be an instance of \cvcAcr{} with~$G=(V,E)$, $V=\{v_1,\dots,v_n\}$, and~$E=\{e_1,\dots,e_m\}$.
 Construct an instance~$I'=(G',\calH,k)$ as follows
 (see~\cref{fig:reductions}(d)).
 Let~$G'=(V',E')$ with $V'\ceq V\cup V^*$ where~$V^*=\{v_1^*,\dots,v_n^*\}$ and~$E'\ceq E\cup E^*\cup\bigcup_{i=1}^n \{\{v_i,v_i^*\}\}$,
 where~$E^*=\{e^*_\ell=\{v_i^*,v_j^*\}\mid e_\ell=\{v_i,v_j\}\in E\}$.
 Let~$\calH=\{H_1,\dots,H_m\}\cup \{H_1^*,\dots,H_m^*\}\cup \{H_1',\dots,H_m'\}$,
 where~$H_i\ceq e_i$, 
 $H^*_i\ceq e_i^*$,
 and~$H_i'\ceq e_i\cup e_i^*$ for all~$i\in\set{m}$.
 Let~$k\ceq 2m+p$.
\end{construction}

\begin{construction}
 \label{constr:bintree}
 Let~$I=(G,p)$ an instance of \cvcAcr{} with~$G=(V,E)$, $V=\{v_1,\dots,v_n\}$, and~$E=\{e_1,\dots,e_m\}$.
 W.l.o.g.\ we assume $n$ to be a power of two as we can add isolated vertices.
 Construct an instance~$I'=(G',\calH,k)$ as follows
 (see~\cref{fig:reductions}(e)).
 Let~$T$ be a full binary tree of height~$\tlog[2]{n}$
 with root~$w$.\footnote{A full binary tree~$T$ with root~$w$ of height~$t$
    is a tree with exactly~$2^t$ leaves each of distance exactly~$t$ to~$w$
    and every inner node has exactly two children.}
 Denote by~$w_1,\dots,w_n$ the leaves in the order provided by a depth-first search starting at~$w$.
 Let~$T'$ be a copy of~$T$,
 and denote by~$w_i'$ the copy of leaf~$w_i$.
 Add~$T$ and~$T'$ to~$G'$ and
 for each~$i\in\set{n}$,
 add the edge~$\{w_i,w_i'\}$.
 For the construction of the habitats,
 denote the non-leaf vertices of~$T$ by
 $w_C$ where~$C\subseteq \set{n}$ is the maximal subset of leaf indices in the subtree of~$T$ rooted at~$w_C$ 
 (analogously for~$T'$).
 For each edge~$e\in E(T)\cup E(T')$,
 $\calH$ contains the habitat~$H_e=e$.
 Now, for each edge~$e_\ell=\{v_i,v_j\}\in E$,
 $\calH$ contains the habitat~$H_\ell'\ceq V(T_\ell) \cup V(T_\ell')$,
 where~$T_\ell$ is the subtree of~$T$ rooted at~$w_C$ with~$C$ being the smallest set with~$\{i,j\}\subseteq C$ 
 (analogously for~$T_\ell'$).
 Finally,
 let~$k\ceq |E(T)|+|E(T')|+p$.
\end{construction}
\paragraph{Correctness.}
We next prove \cref{thm:lbs}(i)--(iv), by using \cref{constr:c3,constr:planar,constr:deg,constr:bintree}
and extending it with the crown (see \cref{def:crown}).

\begin{proof}[Proof of \cref{thm:lbs}(i)]
 Let~$I=(G,p)$ be an instance of CVC,
 and let~$I'=(G',\calH,k)$ be an instance of \RgbpAcr{} obtained from~$I$ using~\cref{constr:c3}.
 We claim that~$I$ is a \yes-instance 
 if and only if
 $I'$ is a \yes-instance.
 
 \RD{}
 Let~$V'\subseteq V$ be a vertex cover of size at most~$p$.
 We claim that~$F=E\cup\bigcup_{v_i\in V'} \{\{x,v_i\}\}$ is a solution to~$I'$.
 Note that~$|F|\leq m+p$ and~$H_i\in F$ for all~$i\in\set{m}$.
 Suppose the claim is false,
 that is,
 there is a habitat~$H_\ell'\in\calH$ that is not connected.
 Since~$e_\ell=\{v_i,v_j\}\subseteq F$,
 neither~$\{x,v_i\}$ nor~$\{x,v_j\}$ is in~$F$.
 Hence,
 $V'\cap e_\ell=\emptyset$,
 a contradiction.
 
 \LD{}
 Let~$F$ be a solution to~$I$.
 We know that~$H_i\in F$ for all~$i\in\set{m}$.
 We claim that~$V'=\{v_i\mid \{x,v_i\}\in F\}$ is a vertex cover of~$G$.
 Note that~$|V'|\leq p$ since~$|F\setminus E|\leq p$.
 Suppose the claim is false,
 that is,
 there is an edge~$e_\ell=\{v_i,v_j\}\in E$ with~$e_\ell\cap V'=\emptyset$.
 By construction of~$V'$,
 we have that each of~$\{x,v_i\}$ and~$\{x,v_j\}$ are not in~$F$.
 Hence the habitat~$H_\ell'$ is not connected,
 a contradiction.
 
 Note that adding any edge to~$G'$ does not change the correctness.
 For the \NP-hardness for only~$C_\ell$-habitats, 
 $\ell\geq 3$,
 replace every edge~$e=\{v,w\}$ in~$E$ by a~$(\ell-3,1)$-crowning
 and adjust~$\calH$ and~$k$ accordingly.
\end{proof}

\begin{proof}[Proof of \cref{thm:lbs}(ii)]
 Let~$I=(G,p)$ be an instance of CVC,
 and let~$I'=(G',\calH,k)$ be an instance of \RgbpAcr{} obtained from~$I$ using~\cref{constr:planar}.
 We claim that~$I$ is a \yes-instance 
 if and only if
 $I'$ is a \yes-instance.
 
 \RD{}
 Let~$V'\subseteq V$ be a vertex cover of size at most~$p$.
 We claim that~$F=\bigcup_{v_i\in V} \{\{x,v_i\}\} \cup \bigcup_{v_i\in V'} \{\{y,v_i\}\}$ is a solution to~$I'$.
 Note that~$|F|\leq n+p$ and~$H_i\in F$ for all~$i\in\set{n}$.
 Suppose the claim is false,
 that is,
 there is a habitat~$H_\ell'\in\calH$ that is not connected.
 Since~$\{\{v_i,x\},\{x,v_j\}\}\subseteq F$,
 neither~$\{y,v_i\}$ nor~$\{y,v_j\}$ is in~$F$.
 Hence,
 $V'\cap e_\ell=\emptyset$,
 a contradiction.
 
 \LD{}
 Let~$F$ be a solution to~$I$.
 We know that~$H_i\in F$ for all~$i\in\set{n}$.
 We claim that~$V'=\{v_i\mid \{y,v_i\}\in F\}$ is a vertex cover of~$G$.
 Note that~$|V'|\leq p$ since~$|F\setminus \bigcup_{v_i\in V} \{\{x,v_i\}\}|\leq p$.
 Suppose the claim is false,
 that is,
 there is an edge~$e_\ell=\{v_i,v_j\}\in E$ with~$e_\ell\cap V'=\emptyset$.
 By construction of~$V'$,
 we have that each of~$\{y,v_i\}$ and~$\{y,v_j\}$ are not in~$F$.
 Hence the habitat~$H_\ell'$ is not connected,
 a contradiction.
 
 For the \NP-hardness for only~$C_\ell$-habitats with even~$\ell\geq 6$ or~$\ell=4$,
 replace every edge~$e=\{v_i,x\}$ by an~$(\ell/2-2,\ell/2)$-crowning.
 For the \NP-hardness for~$C_\ell$-habitats with odd~$\ell\geq 7$,
 replace every edge~$e=\{v_i,x\}$ by two crownings,
 an~$((\ell-1)/2-2,(\ell+1)/2)$-crowning and an~$((\ell+1)/2-2,(\ell-1)/2)$-crowning.
 Adjust~$\calH$ and~$k$ accordingly.
\end{proof}

\begin{proof}[Proof of \cref{thm:lbs}(iii)]
 Let~$I=(G,p)$ be an instance of CVC,
 and let~$I'=(G',\calH,k)$ be an instance of \RgbpAcr{} obtained from~$I$ using~\cref{constr:deg}.
 We claim that~$I$ is a \yes-instance 
 if and only if
 $I'$ is a \yes-instance.
 
 \RD{}
 Let~$V'\subseteq V$ be a vertex cover of size at most~$p$.
 We claim that~$F=E\cup E^*\cup \bigcup_{v_i\in V'} \{\{v_i^*,v_i\}\}$ is a solution to~$I'$.
 Note that~$|F|\leq 2m+p$ and~$H_i,H_i^*\in F$ for all~$i\in\set{m}$.
 Suppose the claim is false,
 that is,
 there is a habitat~$H_\ell'\in\calH$ that is not connected.
 Since~$\{\{v_i,v_j\},\{v_i^*,v_j^*\}\}\subseteq F$,
 neither the edges~$\{v_i^*,v_i\}$ nor~$\{v_j^*,v_j\}$ are in~$F$.
 Hence,
 $V'\cap e_\ell=\emptyset$,
 a contradiction.
 
 \LD{}
 Let~$F$ be a solution to~$I$.
 We know that~$H_i,H_i^*\in F$ for all~$i\in\set{m}$.
 We claim that~$V'=\{v_i\mid \{v_i,v_i^*\}\in F\}$ is a vertex cover of~$G$.
 Note that~$|V'|\leq p$ since~$|F\setminus (E\cup E^*)|\leq p$.
 Suppose the claim is false,
 that is,
 there is an edge~$e_\ell=\{v_i,v_j\}\in E$ with~$e_\ell\cap V'=\emptyset$.
 By construction of~$V'$,
 we have that each of~$\{v_i^*,v_i\}$ and~$\{v_j^*,v_j\}$ are not in~$F$.
 Hence the habitat~$H_\ell'$ is not connected,
 a contradiction.
 
 For the \NP-hardness for~$\{P_2,C_\ell\}$-habitats with~$\ell\geq 4$,
 subdivide each edge~$\ell-4$ times.
 For the \NP-hardness for only~$C_\ell$-habitats with~$\ell\geq 4$,
 replace every edge~$e\in E$ by an~$(\ell-4,2)$-crowning
 and every edge~$e^*\in E^*$ by an~$(0,\ell-2)$-crowning.
 Note that this increases the maximum degree by six to at most~ten.
 Adjust~$\calH$ and~$k$ accordingly.
\end{proof}

\begin{proof}[Proof of \cref{thm:lbs}(iv)]
 Let~$I=(G,p)$ be an instance of \vcAcr{},
 and let~$I'=(G',\calH,k)$ be an instance of \RgbpAcr{} obtained from~$I$ using~\cref{constr:bintree}.
 It is not difficult to see that~$\Delta(G)\leq 3$ and
 that each habitat induces either a~$P_2$ or a cycle.
 We claim that~$I$ is a \yes-instance 
 if and only if
 $I'$ is a \yes-instance.
 For notation,
 let~$E_T\ceq E(T)\cup E(T')$.
 
 \RD{}
 Let~$V'\subseteq V$ be a vertex cover of size at most~$p$.
 We claim that~$F=\bigcup_{e\in E_T} \{e\} \cup \bigcup_{v_i\in V'} \{\{w_i,w_i'\}\}$ is a solution to~$I'$.
 Note that~$|F|\leq |E_T|+p$ and~$H_e\in F$ for all~$e\in E_T$.
 Suppose the claim is false,
 that is,
 there is a habitat~$H_\ell'\in\calH$ corresponding to~$e_\ell=\{v_i,v_j\}$ that is not connected.
 Thus,
 each of the edges~$\{w_i,w_i'\}$ and~$\{w_j,w_j'\}$ is not in~$F$.
 Hence,
 $V'\cap e_\ell=\emptyset$,
 a contradiction.
 
 \LD{}
 Let~$F$ be a solution to~$I$.
 We know that~$H_e\in F$ for all~$e\in E_T$.
 We claim that~$V'=\{v_i\mid \{w_i,w_i'\}\in F\}$ is a vertex cover of~$G$.
 Note that~$|V'|\leq p$ since~$|F\setminus E_T|\leq p$.
 Suppose the claim is false,
 that is,
 there is an edge~$e_\ell=\{v_i,v_j\}\in E$ with~$e_\ell\cap V'=\emptyset$.
 By construction of~$V'$,
 we have that each of~$\{w_i,w_i'\}$ and~$\{w_j,w_j'\}$ are not in~$F$.
 Hence the habitat~$H_\ell'$ is not connected,
 a contradiction.
 
 For the \NP-hardness for only cycle-habitats,
 replace every edge~$e\in E_T$ by a~$(0,1)$-crowning on~$e$'s endpoints.
 Note that this increases the maximum degree by~six to at most~nine.
 Adjust~$\calH$ and~$k$ accordingly.
\end{proof}
\section{Upper Bounds}
\label{sec:ubs}
\ifshort{}\appendixsection{sec:ubs}\fi{}

This section is devoted to instances of \RwgbpAcr{} in which every habitat induces a cycle.
We will first show that this case can be reduced to the following problem.

\decprob{\mchmTsc{} (\mchmAcr{})}{mchm}
{A hypergraph~$G=(V,E)$ with edge weights~$w\colon E\to \N$.}
[Task]{Find a set~$M\subseteq E$ of maximum weight such that for all~$e,e'\in M$ holds that~$e\cap e'=\emptyset$.}

\noindent
\mchmAcr{} is \NP-hard \cite{GareyJ79}, but if every hyperedge is of cardinality at most two,
it is equivalent to the well-known \mcmTsc{} (\mcmAcr{}) problem which is solvable in~$O(|V|(|E|+\log |V|))$ time \cite{Gabow90}.
We make use of this to prove that some special cases of \RwgbpAcr{} are polynomial-time solvable.

\subsection{The general case for cycles}
In this subsection we show the following.

\begin{proposition}%
 \label{prop:habindface}
 \RwgbpAcr{} where every habitat induces a cycle 
 can be decided by solving \mchmAcr{} 
 where the largest hyperedge is of size of the largest number of habitats intersecting in one edge.
\end{proposition}

\begin{remark}
\mchmAcr{} admits an ILP formulation with 
linearly many variables and constraints
(used for our experiments):
\begin{equation}
\begin{aligned}\label{eq:A}
 \max && \sum_{e\in E(H)} c(e)\cdot x_e \\
 \text{s.t.} && x_e&\in\{0,1\} && \forall e\in E(H) \\
 &&  \sum_{e\in E(H):\: v\in e} x_e &\leq 1 && \forall v\in V(H)
\end{aligned}
\end{equation}
\end{remark}

\noindent
Central for the translation to \mchmAcr{} is the following graph.

\begin{definition}[Habitat graph]
\label{def:habgraph}
Let~$G=(V,E)$ be a graph with edge cost~$c\colon E\to \N$ and~$\calH$ be a set of habitats each inducing a cycle.
The multi-hypergraph~$B=(V_B,E_B)$ with edge weights~$w_B\colon E_B\to\N$ and bijection~$f:E\to E(B)$ are obtained as follows.
$B$ contains a vertex~$b_i$ for each habitat~$H_i$.
For every edge~$e\in E$ shared by at least two habitats~$H_{i_1},\dots,H_{i_j}$,
add a hyperedge~$e'=\{b_{i_1},\dots,b_{i_j}\}$
and set~$f(e)\ceq e'$ and~$w_B(e')\ceq c(e)$. 
Finally,
for every edge~$e\in E$ induced by only one habitat~$H_i$,
add a vertex~$b_e$ and the edge~$e'\ceq \{b_e,b_i\}$,
and set~$f(e)\ceq e'$ and~$w_B(e')\ceq c(e)$.
\end{definition}

\noindent
The following connection between \RwgbpAcr{} 
and \mcmAcr{} proves~\cref{prop:habindface}.

\begin{lemma}
 \label{lem:matchsol}
 Let~$G=(V,E)$ be a graph with edge cost $c\colon E\to \N$, let~$\calH$ be a set of habitats each inducing a cycle,
 and let~$B$ denote the habitat graph with edge weights $w_B\colon E_B\to\N$ and function~$f$.
 \begin{inparaenum}[(i)]
 \item If $M$ is a matching in~$B$,
 then
 $G[E\setminus f^{-1}(M)][H]$ is connected for every~$H\in\calH$;
 \item If $G[F][H]$ is connected for every~$H\in\calH$,
 then~$E(B)\setminus f(F)$ is a matching in~$B$.
 \end{inparaenum}
\end{lemma}

\begin{proof}
 (i) 
 Let~$M$ be a matching in~$B$ and let~$F\ceq E\setminus f^{-1}(M)$. 
 Then for every~$b_i\in V(B)$ there is at most one edge in~$M$ that is incident with~$b_i$.
 Thus,
 for every~$H\in \calH$,
 $|E(G[H])\cap F|\geq |H|-1$, 
 and hence~$G[F][H]$ is connected.
 
 (ii)
 Let $G[F][H]$ be connected for every~$H\in\calH$
 and let~$M\ceq E(B)\setminus f(F)$.
 Suppose there are two edges in~$M$ that are both incident to some~$b_i\in V(B)$.
 Then there are two edges in~$E(G[H_i])$ not contained in~$F$,
 and hence~$G[H_i][F]$ is not connected---a contradiction.
 Thus,
 $M$ is a matching.
\end{proof}

\appendixproof{prop:habindface}
{
\begin{proof}[Proof of~\cref{prop:habindface}]
 Due to~\cref{lem:matchsol}(i),
 we know that every matching forms a solution.
 With the addition of~\cref{lem:matchsol}(ii),
 we know that every maximum-weight matching forms a minimum-cost solution.
\end{proof}
}

\begin{remark}
 We can simplify the habitat graph to a simple hypergraph:
 If there are multiple edges with the same vertex set,
 then it is enough to keep exactly one of maximum weight.
 We will make use of this in our experiments.
\end{remark}

\subsection{Polynomial-time solvable subcases}

If every habitat induces a cycle
and every edge is in at most two habitats, 
then the 
habitat graph
is a hypergraph with edges of cardinality at most two.
We have the following.

\begin{theorem}
 \label{thm:habindface}
 \RwgbpAcr{} where every habitat induces a cycle is solvable in~$O(|V|\cdot |E|\cdot |\calH|)$ time
 when every edge is in at most two habitats.
\end{theorem}

\noindent
We next present two special cases of \RwgbpAcr{} that become polynomial-time solvable due to the above.

\paragraph{Habitats inducing faces.}

Suppose our input graph is a plane graph 
(that is, a planar graph together with an crossing-free embedding into the plane).
If every habitat induces a cycle which is the boundary of a face,
then clearly every edge is shared by at most two habitats since
every edge is incident with exactly two faces.
Thus, we get the following.

\begin{corollary}
 \label{cor:habindface}
 \RwgbpAcr{} where every habitat induces a cycle is solvable in~$O(|V|\cdot |E|\cdot |\calH|)$ time
 on plane graphs 
 when every habitat additionally induces a face.
\end{corollary}

\paragraph{Habitats inducing triangles in graphs of maximum degree three.}

Suppose our input graph has maximum degree three and each habitat induces a triangle.
Observe that every vertex of degree at most one cannot be contained in a habitat.
Moreover,
every degree-two vertex is contained in at most one habitat.
For degree-three vertices we have the following.

\begin{lemma}%
 \label{lem:Kfour}
 If a vertex~$v$ is contained in three habitats,
 then~$N_G[v]$ is a connected component isomorphic to a~$K_4$.
\end{lemma}

\appendixproof{lem:Kfour}
{
\begin{proof}
 Firstly,
 observe that~$6$ slots are distributed among 3 vertices,
 and hence,
 for every~$w\in N(v)$
 there are two distinct habitats~$H, H'$ with~$\{v,w\}\subseteq H\cap H'$.
 This means that~$|N(w)\cap N(v)|=2$.
 Hence,
 each vertex~$G[N[v]]$ has degree three,
 and thus $N_G[v]$ is a connected component.
\end{proof}
}

\noindent
We immediately derive the following data reduction rule.

\begin{rrule}
If a vertex~$v$ is contained in three habitats,
then delete~$N_G[v]$ and reduce~$k$ by the minimum cost of a solution for~$G[N_G[v]]$.
\end{rrule}

\noindent
If the reduction rule is inapplicable,
then every vertex is contained in at most two habitats.
Consequently, every edge is shared by at most two habitats.
We obtain the following.

\begin{corollary}
 \label{cor:deg3}
 \RgbpAcr{} on graphs of maximum degree three is solvable in~$O(|V|\cdot |E|\cdot |\calH|)$ time
 when every habitat induces a triangle.
\end{corollary}

\section{Experiments}
\label{sec:experiments}
\ifshort{}\appendixsection{sec:experiments}\fi{}

In this section,
we present and discuss our experimental and empirical evaluation.
We explain our data in~\cref{ssec:data},
our algorithms in~\cref{ssec:algs},
and our results in~\cref{ssec:results}.

\subsection{Data}
\label{ssec:data}

\paragraph{Graphs.}
Our experiments are conducted on planar graphs only.
We used data freely available by Open Street Maps (OSM).
For each state of Germany (except for the city states Berlin, Bremen and Hamburg;
abbreviated by {ISO} 3166 code),
we set a bounding box and extracted the highways within.
For each area encapsulated by highways, we created a vertex.
Two vertices are connected by an edge whenever they are adjacent by means of a highway.
Additionally,
we generated five artifical graphs which are relative neighborhood graphs~\cite{Toussaint80}
of sets of $500+i\cdot1125$~points in the plane, placed uniformly at random,
with~$i\in\set[0]{4}$.
To all graphs we randomly assigned edge costs from~$\{1,\dots,8\}$.
\cref{tab:summary_prop}
provides an overview over some instances' properties.

\begin{table}[t]\centering
 \caption{Properties of our real-world and artificial graphs.}
 \begin{tabular}{@{}r|rrrrrr@{}}\toprule
 & $|V|$ & $|E|$ & $|$faces$|$ & $\delta$ & $2|E|/|V|$ & $\Delta$ \\\midrule\midrule
BB & 129 & 293 & 165 & 2 & 4.543 & 11\\
BW & 224 & 508 & 285 & 1 & 4.536 & 15\\
BY & 284 & 651 & 368 & 1 & 4.585 & 11\\
HE & 113 & 245 & 133 & 1 & 4.336 & 13\\
MV & 74 & 159 & 86 & 2 & 4.297 & 9\\
NI & 141 & 321 & 181 & 2 & 4.553 & 11\\
NW & 455 & 1035 & 581 & 1 & 4.549 & 14\\
RP & 138 & 295 & 158 & 1 & 4.275 & 13\\
SH & 49 & 98 & 50 & 1 & 4 & 8\\
SL & 65 & 133 & 69 & 1 & 4.092 & 12\\
SN & 96 & 211 & 116 & 1 & 4.396 & 10\\
ST & 133 & 302 & 170 & 1 & 4.541 & 10\\
TH & 94 & 205 & 112 & 1 & 4.362 & 9\\
\midrule
A500 & 500 & 629 & 130 & 1 & 2.516 & 4\\
A1625 & 1625 & 2038 & 414 & 1 & 2.508 & 4\\
A2750 & 2750 & 3458 & 709 & 1 & 2.515 & 4\\
A3875 & 3875 & 4871 & 997 & 1 & 2.514 & 4\\
A5000 & 5000 & 6315 & 1316 & 1 & 2.526 & 4\\
\bottomrule
 \end{tabular}
 \label{tab:summary_prop}
\end{table}

\paragraph{Habitats.}
We created multiple instances from every graph above by equipping them with different types and numbers of habitats.
We created three types of instances: \emph{face} instances, \emph{cycle} instances, and \emph{walk} instances.
Given a plane graph $G$, a number $r$ of habitats and, in the case of cycle and walk instances, a habitat size~$q$, the instances were created as follows.
\begin{asparadesc}
 \item[\rm\it Face instances:] Out of those faces of $G$ that induce cycles, randomly choose~$r$ faces as habitats.
 \item[\rm\it Cycle instances:]
List all induced cycles of length~$q \pm 1$.
For each such~$q$, randomly choose~$r$ of the cycles as habitats.
 \item[\rm\it Walk instances:]
Compute $r$ self-avoiding random walks on~$q' = q\pm1$ vertices, where~$q'$ is chosen uniformly at random.
Add the vertices of each walk to a habitat.
\end{asparadesc}

\noindent
For each instance type,
for each graph above,
for each $r \in \{50, 100, 150, 200\}$,
for each $q \in \{5, 7, \dots, 13\}$ in the case of cycle and walk instances,
we generated 5 instances.

We remark that the real-world graphs MV, SH, and SL did not have sufficiently many cycles of length~$q\pm1$ for~$q\in\{5,7\}$, $q\in\{5, 7, 9\}$, and $q=5$, respectively.
In this case, every cycle was chosen to be a habitat.

\cref{fig:osmgraph} is a drawing of the graph SL, based on the street network of Saarland, together with a set of cycle habitats.
\begin{figure}[t]
	\centering
	\def\svgwidth{\columnwidth}
	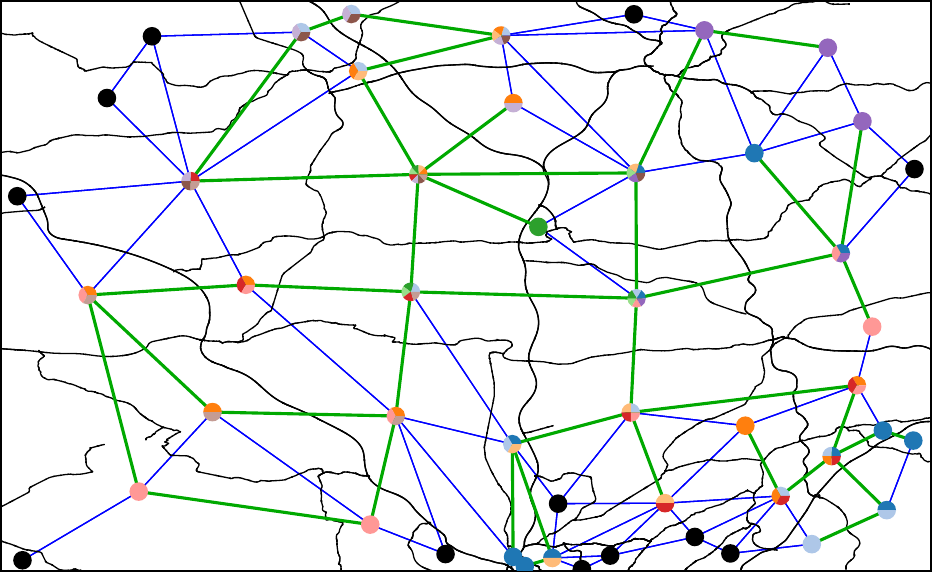
	\caption{Drawing of the graph SL (blue and green edges), based on the street network of Saarland (black).
		The graph is equipped with 21 cycle habitats, each of which is represented by a unique color.
	The green edges represent an optimum solution for the instance.}
	\label{fig:osmgraph}
\end{figure}

\subsection{Algorithms}
\label{ssec:algs}

\def\Amcm{\calA_{\mathrm{mwm}}}
\def\Amchm{\calA_{\mathrm{mwhm}}}
\def\Aapx{\calA_{\mathrm{apx}}}
\def\Agen{\calA_{\mathrm{gen}}}

\begin{table*}[t]\centering
 \caption{Summary of our results regarding~$\Aapx$.
 The quality ratio is cost($\Aapx$)/OPT.
 The additive ratio is (cost($\Aapx$)$-$OPT)/($d\cdot|\calH|$),
 where~$d$ is the average weight of an edge in an optimal solution.
 The running time ratio is time(BEST)/time($\Aapx$),
 where BEST is the exact solver with the best overall running time on the instance.}
 \begin{tabular}{@{}r|rrrr|rrrr|rrrr@{}}\toprule
  & \multicolumn{4}{c|}{quality ratio} & \multicolumn{4}{c|}{additive ratio} & \multicolumn{4}{c}{running time ratio} \\
  & min & max & mean & sd & min & max & mean & sd & min & max & mean & sd \\
  \midrule\midrule
  Faces$_{\mathrm{Art}}$ & 1 & 1.134 & 1.039 & 0.036 & 0 & 0.575 & 0.125 & 0.178 & 0.586 & 10.877 & 5.142 & 2.884\\
Cycles$_{\mathrm{Art}}$ & 1 & 1.119 & 1.038 & 0.026 & 0 & 0.724 & 0.234 & 0.151 & 0.583 & 19.957 & 8.265 & 4.694\\
Walk$_{\mathrm{Art}}$ & 1 & 1.044 & 1.008 & 0.009 & 0 & 0.31 & 0.053 & 0.056 & 2.835 & 4652.877 & 737.436 & 1157.331\\
\midrule
Faces$_{\mathrm{Real}}$ & 1.018 & 1.243 & 1.141 & 0.049 & 0.042 & 0.326 & 0.201 & 0.055 & 0.802 & 5.586 & 2.256 & 1.015\\
Cycles$_{\mathrm{Real}}$ & 1.026 & 1.313 & 1.16 & 0.041 & 0.03 & 0.853 & 0.29 & 0.147 & 1.515 & 22.083 & 4.648 & 2.105\\
Walk$_{\mathrm{Real}}$ & 1.016 & 1.348 & 1.174 & 0.054 & 0.007 & 1.27 & 0.34 & 0.233 & 5.456 & 5797.882 & 1266.843 & 1815.414\\
  \bottomrule
 \end{tabular}
 \label{tab:sumamry_apx}
\end{table*}

\begin{figure*}[t]
    \begin{subfigure}[c]{0.395\textwidth}
        \includegraphics[width=1\textwidth]{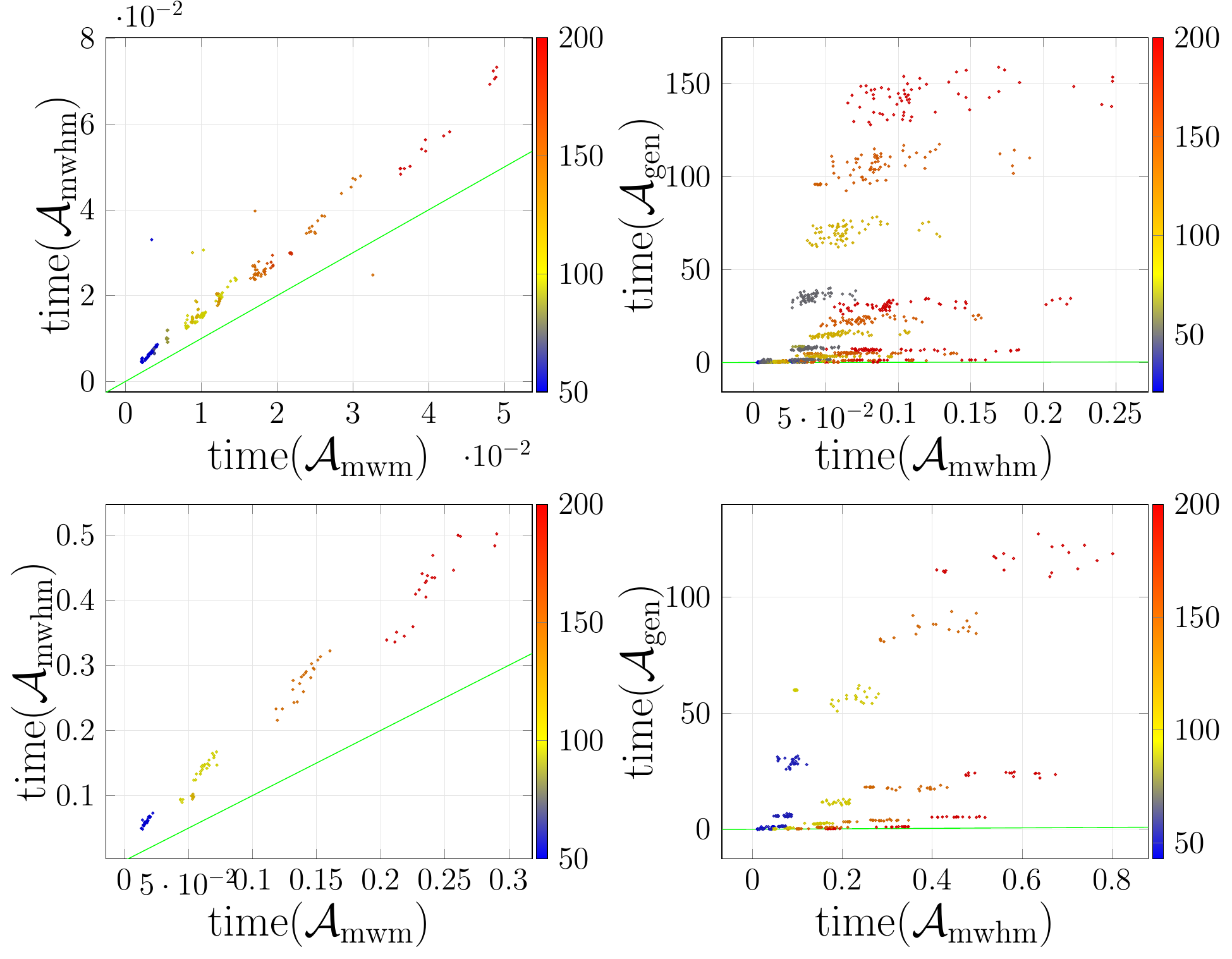}
        \subcaption{
        (Left) Faces.
        (Right) Cycles.}
    \end{subfigure}
    \hfill
    \begin{subfigure}[c]{0.595\textwidth}
        \includegraphics[width=1\textwidth]{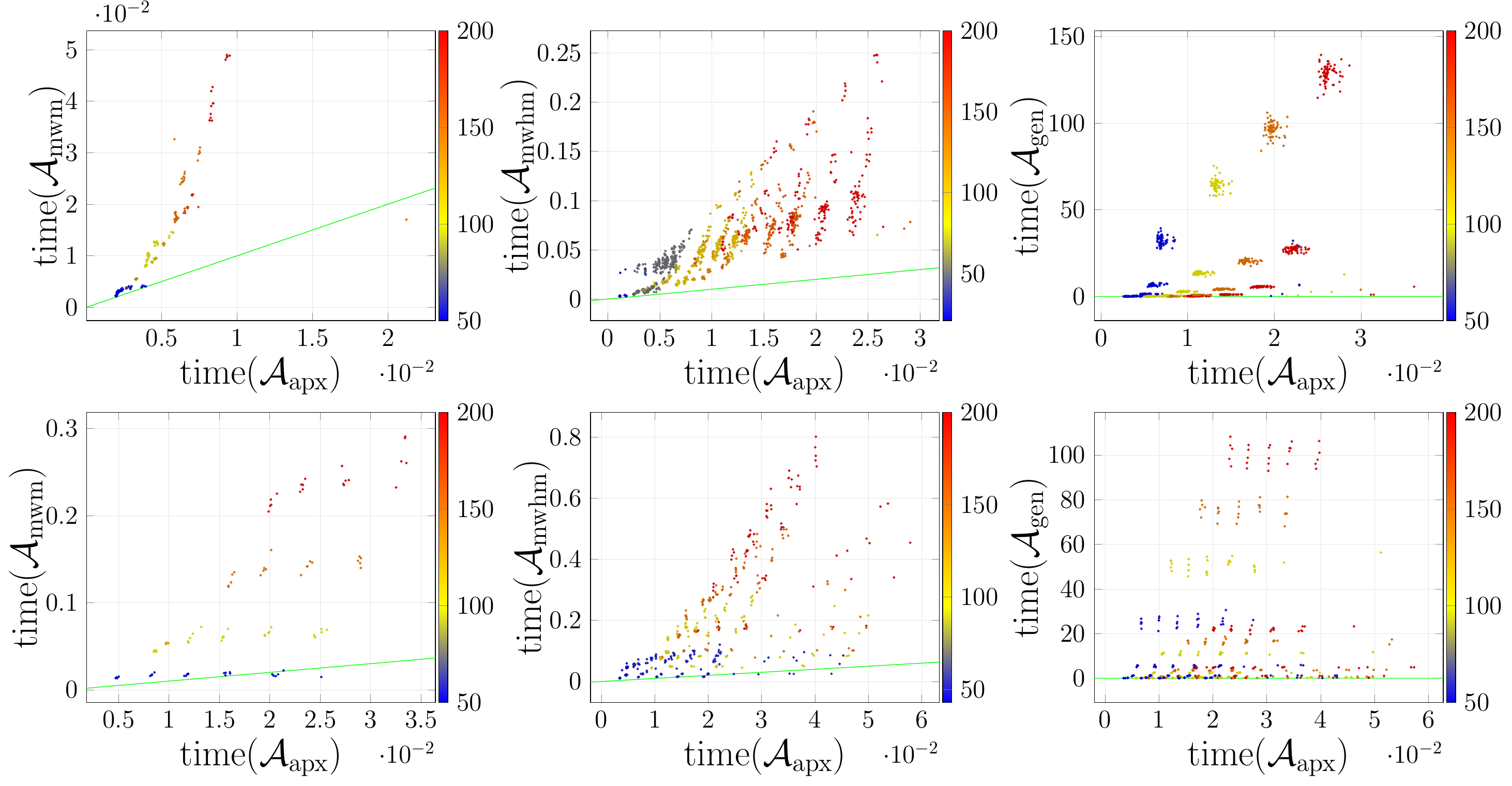}
        \subcaption{
        (Left) Faces.
        (Middle) Cycles.
        (Right) Walks.}
    \end{subfigure}
    \caption{(Top) Real-World.
    (Bottom) Artificial.
    Data points are colored by the number of habitats.
    (a) Running times of the two best exact solvers on the respective instance type.
    (b) Running times of~$\Aapx$ against the best exact solver on the respective instance type.
    }
    \label{fig:Best}
\end{figure*}

We implemented three exact solvers and one approximate solver.
Two of the three exact solvers can be run only on some types of habitats
(see \cref{tab:algs}).
\toappendix{
\def\Cmark{\checkmark}
\begin{table}[t]
 \centering
 \caption{Overview on algorithms and application domains.}
 \begin{tabular}{@{}rllll@{}}\toprule
    & $\Amcm$ & $\Amchm$ & $\Agen$ & $\Aapx$ \\\midrule\midrule
  Face habitats & \Cmark & \Cmark & \Cmark & \Cmark \\
  Cycle habitats & \xmark & \Cmark & \Cmark & \Cmark \\
  Walk habitats & \xmark & \xmark & \Cmark & \Cmark \\
  \bottomrule
 \end{tabular}
 \label{tab:algs}
\end{table}
}

For the face instances we implemented the \mcmAcr{}-based algorithm~(\cref{cor:habindface}) which we will denote by~$\Amcm$.
The habitat graph generation is implemented in Python 3. 
The matching is computed using Kolmogorov's~\cite{Kolmogorov09} C++ implementation of the Blossom V algorithm.

For the cycle instances generated the habitat graph in Python 3 and used Gurobi 9.5.0 to solve ILP formulation~\eqref{eq:A}.
We call the respective solver~$\Amchm$.

The generic solver~$\Agen$ can solve all instances of \RwgbpAcr{} and uses Gurobi 9.5.0 to solve the following ILP formulation
with an exponential number of constraints.
\begin{equation*}
\begin{aligned}
 \min && \sum_{e\in E(G)} c(e)\cdot x_e \\
 \text{s.t.} && x_e&\in\{0,1\} && \forall e\in E(G) \\
 && \sum_{e\in \delta_H(S)} x_e &\geq 1 && \forall \emptyset\neq S\subsetneq H,\, \forall H\in\calH
\end{aligned}
\end{equation*}
By~$\delta_H(S)\ceq \{e\in E\mid e\cap S\neq \emptyset \land e\cap (H\setminus S)\neq \emptyset\}$,
we denote for a graph~$G = (V, E)$,
vertex set~$H\subseteq V$,
and subset~$\emptyset\neq S\subsetneq H$,
the set of edges between~$S$ and~$H\setminus S$.

The approximate solver~$\Aapx$, implemented in Python 3, is a weighted adaption of the $O(r)$-approximation algorithm for \RgbpAcr{} given by~\citet{FK21}.
Their algorithm computes for every habitat a spanning tree,
and then combines the solutions.
The weighted adaption has the same approximation guarantee.
Further, for induced cycles it has an additive approximation guarantee that depends on the number of habitats and the maximum cost of any edge.

\begin{observation}%
 \label{obs:approxguarantee}
 $\Aapx$ is an additive~$(r \cdot c_\mathrm{max})$-ap\-prox\-i\-ma\-tion for \RwgbpAcr{}
 with each habitat inducing a cycle,
 where~$c_\mathrm{max} = \max_{e \in E} c(e)$.
\end{observation}

\appendixproof{obs:approxguarantee}
{
\begin{proof}
 Let~$F_1,\dots,F_r$ and~$F\ceq \bigcup_{i=1}^r F_i$ denote the solution of the approximation algorithm,
 and let~$F^*=\bigcup_{i=1}^r F_i^*$ denote any optimal solution.
 For each~$i\in\set{r}$,
 let~$\{e_i'\}=F_i\setminus F_i^*$.
 Then:
 \begin{align*}
  c(F) &= c(F_1\cup\dots\cup F_r) 
  \\
  &\leq c(F_1^*\cup\{e_1'\}\cup\dots\cup F_r^*\cup\{e_r'\})
  \\
  &\leq c(F_1^*\cup\dots\cup F_r^*) + \sum_{i=1}^r c(e_i') 
  \\
  &\leq c(F^*)+r\cdot c_\mathrm{max}. \qedhere
 \end{align*}
\end{proof}
}
\subsection{Results}
\label{ssec:results}

We compared the implementations on machines equipped with an Intel Xeon W-2125 CPU and 256GB of RAM running Ubuntu 18.04.
For the ILP-based solvers, we set a time limit of 30s for the solving time (not the build time).
For 43 of the 100 artificial faces instances $\Agen$ was not able to compute any feasible solution. 
For all remaining instances, 
$\Agen$ provided an optimal solution in the given time limit.

All material to reproduce the results is provided in the supplementary material.

\paragraph{Comparison of the optimal solvers.}

Our experiments underline 
that each specialized solver outperforms
the next less specialized solver
(see~\cref{fig:Best}(a)).
For instance,
on real-world instances with face habitats,
$\Amcm$ is on average $1.7$ times faster than~$\Amchm$,
and on artificial instances with cycle habitats,
$\Amcm$ is on average $82$ times faster than~$\Agen$
(see~\cref{tab:summary_best}).
\begin{table}[t]\centering
 \caption{Summary of our results regarding the running time ratio of the respective two best exact algorithms.}
 \begin{tabular}{@{}r|rrrr@{}}\toprule
  & \multicolumn{4}{c}{running time ratio} \\
  & min & max & mean & sd \\
  \midrule\midrule
  Faces$_{\mathrm{Art}}$ & 1.578 & 4.424 & 2.382 & 0.709\\
  Faces$_{\mathrm{Real}}$ & 0.76 & 9.566 & 1.756 & 0.572\\
  \midrule
  Cycles$_{\mathrm{Art}}$ & 1.35 & 645.87 & 82.933 & 137.271\\
  Cycles$_{\mathrm{Real}}$ & 0.834 & 2266.312 & 315.597 & 483.523\\
  \bottomrule
 \end{tabular}
 \label{tab:summary_best}
\end{table}
Moreover,
$\Amcm$ is $1.5$~times faster than~$\Amchm$ on 80\% of the face instances,
and $\Amchm$ is $10$~times faster than~$\Agen$ on 76\% of the cycle instances.

\paragraph{Approximate solver.}
On real-world instances, $\Aapx$ is
$2$ times faster than $\Amcm$ on face instances and
$4$ times faster than $\Amchm$ on cycle instances,
whereas the approximation factor never exceeds~$1.348$.
The additive error is
significantly better than the theoretical bound in~\cref{obs:approxguarantee}.
See \cref{tab:sumamry_apx} and~\cref{fig:Best}(b).
On the artificial instances,
the approximation ratios are even better on average.
This may be due to the fact that these instances are rather sparse
(see~\cref{tab:summary_prop}).

\smallskip
\noindent\emph{Intersection.}
We also considered the intersection rate $\lambda \ceq |\sum_{H\in\calH} |E(G[H])|/|\bigcup_{H\in\calH} |E(G[H])|$, 
which measures the average number of habitats in each edge
(see \cref{fig:apx-cost-vs-intersect,fig:apx-time-vs-intersect} in the appendix for a comparison of the solution quality and the speedup factor of $\Aapx$).
For~$\lambda \ge 10$ the approximation quality improves slightly in the real-world cycle and walk instances.
As the habitats lie more dense, it is more likely for an edge to be in the solution.
It thus seems plausible that~$\Aapx$ chooses fewer unnecessary edges.

As for the running time, the intersection rate seems only to have an effect on $\Amchm$.
Especially on real-world cycle instances one can see that the running time quotient of $\Amchm$ and~$\Aapx$ decreases with growing~$\lambda$.
This is likely due to the habitat graph~$B$:
If there are less edges that are in a unique habitat,
then there are less vertices in $B$ that are incident to a single cardinality-two edge,
and hence less variables and constraints in~\eqref{eq:A}.
Note that one can deal with such vertices in a preprocessing routine as proposed i.e.\ by~\citet{KKNNZ21}.
This may significantly improve the running times of~$\Amcm$ and~$\Amchm$ for instances with low intersection rate.

\section{Conclusion}

While we prove that when every habitat induces a cycle,
\RwgbpAcr{} remains \NP-hard, even on planar graphs and graphs of small maximum degree,
we provide an ILP-based solver that 
performs exceptionally well in our experiments.
Moreover, when each habitat additionally induces a face in a given plane graph,
the problem becomes solvable in polynomial time,
with even faster practical running times.
Lastly we observe that
the approximation algorithm by~\citet{FK21}
performs well in our experiments.

In a long version of this paper, 
we wish to address several theoretical and experimental tasks.
On the experimental side,
we plan to test our code on larger input instances.
As mentioned in \cref{ssec:results}, we believe that the implementation of perprocessing routines may improve the running times of $\Amcm$ and~$\Amchm$ significantly.
On the theoretical side,
we plan to settle the computational complexity of \RwgbpAcr{} 
for the following cases:
for habitats in~$\{P_2,C_3\}$ in planar graphs or with constant maximum degree at least~$4$;
for habitats in~$\{P_2,C_\ell\}$ in graphs of maximum degree at most~$3$.
Future work
may include restrictions on the habitats other than induced cycles.

{\begingroup
  \let\clearpage\relax
  \renewcommand{\url}[1]{\href{#1}{$\ExternalLink$}}
  \renewcommand{\doi}[1]{}
  \bibliography{gbp-bib}

\newcommand{\noopsort}[1]{}
\begin{thebibliography}{23}
\providecommand{\natexlab}[1]{#1}
\providecommand{\url}[1]{\texttt{#1}}
\expandafter\ifx\csname urlstyle\endcsname\relax
  \providecommand{\doi}[1]{doi: #1}\else
  \providecommand{\doi}{doi: \begingroup \urlstyle{rm}\Url}\fi

\bibitem[Ament et~al.(2014)Ament, Callahan, McClure, Reuling, and
  Tabor]{Ament15}
R.~Ament, R.~Callahan, M.~McClure, M.~Reuling, and G.~Tabor.
\newblock Wildlife connectivity: Fundamentals for conservation action, 2014.
\newblock Center for Large Landscape Conservation: Bozeman, Montana.

\bibitem[Bastille-Rousseau et~al.(2018)Bastille-Rousseau, Wall,
  Douglas-Hamilton, and Wittemyer]{BastilleWDW18}
Guillaume Bastille-Rousseau, Jake Wall, Iain Douglas-Hamilton, and George
  Wittemyer.
\newblock Optimizing the positioning of wildlife crossing structures using gps
  telemetry.
\newblock \emph{Journal of Applied Ecology}, 55\penalty0 (4):\penalty0
  2055--2063, 2018.
\newblock URL \url{https://doi.org/10.1111/1365-2664.13117}.

\bibitem[Bennett(2017)]{bennett2017effects}
Victoria~J Bennett.
\newblock Effects of road density and pattern on the conservation of species
  and biodiversity.
\newblock \emph{Current Landscape Ecology Reports}, 2\penalty0 (1):\penalty0
  1--11, 2017.

\bibitem[Clevenger et~al.(2002)Clevenger, Wierzchowski, Chruszcz, and
  Gunson]{clevenger2002gis}
Anthony~P Clevenger, Jack Wierzchowski, Bryan Chruszcz, and Kari Gunson.
\newblock Gis-generated, expert-based models for identifying wildlife habitat
  linkages and planning mitigation passages.
\newblock \emph{Conservation biology}, 16\penalty0 (2):\penalty0 503--514,
  2002.

\bibitem[Downs et~al.(2014)Downs, Horner, Loraamm, Anderson, Kim, and
  Onorato]{DownsHLAKO14}
Joni~A. Downs, Mark~W. Horner, Rebecca~W. Loraamm, James Anderson, Hyun Kim,
  and Dave Onorato.
\newblock Strategically locating wildlife crossing structures for {F}lorida
  panthers using maximal covering approaches.
\newblock \emph{Trans. {GIS}}, 18\penalty0 (1):\penalty0 46--65, 2014.
\newblock URL \url{https://doi.org/10.1111/tgis.12005}.

\bibitem[Fluschnik and Kellerhals(2021)]{FK21}
Till Fluschnik and Leon Kellerhals.
\newblock Placing green bridges optimally, with a multivariate analysis.
\newblock In \emph{Proceedings of the 17th Conference on Computability in
  Europe -- Connecting with Computability (CiE~'21)}, pages 204--216, 2021.
\newblock URL \url{https://doi.org/10.1007/978-3-030-80049-9\_19}.

\bibitem[Gabow(1990)]{Gabow90}
Harold~N. Gabow.
\newblock Data structures for weighted matching and nearest common ancestors
  with linking.
\newblock In \emph{Proceedings of the 1st Symposium on Discrete Algorithms
  (SODA~'90)}, pages 434--443, 1990.
\newblock URL \url{http://dl.acm.org/citation.cfm?id=320176.320229}.

\bibitem[Galpern et~al.(2011)Galpern, Manseau, and Fall]{galpern2011patch}
Paul Galpern, Micheline Manseau, and Andrew Fall.
\newblock Patch-based graphs of landscape connectivity: a guide to
  construction, analysis and application for conservation.
\newblock \emph{Biological conservation}, 144\penalty0 (1):\penalty0 44--55,
  2011.

\bibitem[Garey and Johnson(1979)]{GareyJ79}
M.~R. Garey and David~S. Johnson.
\newblock \emph{Computers and Intractability: {A} Guide to the Theory of
  NP-Completeness}.
\newblock W. H. Freeman, 1979.
\newblock ISBN 0-7167-1044-7.

\bibitem[Goldingay and Taylor(2017)]{goldingay2017can}
Ross~L Goldingay and Brendan~D Taylor.
\newblock Can field trials improve the design of road-crossing structures for
  gliding mammals?
\newblock \emph{Ecological Research}, 32\penalty0 (5):\penalty0 743--749, 2017.

\bibitem[Hamer and McDonnell(2008)]{HammerM08}
Andrew~J. Hamer and Mark~J. McDonnell.
\newblock Amphibian ecology and conservation in the urbanising world: A review.
\newblock \emph{Biological Conservation}, 141\penalty0 (10):\penalty0
  2432--2449, 2008.
\newblock ISSN 0006-3207.
\newblock URL \url{https://doi.org/10.1016/j.biocon.2008.07.020}.

\bibitem[Hennings and Soll(2010)]{hennings2010wildlife}
Lori~A Hennings and Jonathan~Andrew Soll.
\newblock \emph{Wildlife corridors and permeability: a literature review}.
\newblock Metro Sustainability Center, 2010.

\bibitem[Huijser et~al.(2008)Huijser, McGowen, Hardy, Kociolek, Clevenger,
  Smith, and Ament]{HuijserMHKCSA08}
Marcel~P. Huijser, Pat~T. McGowen, Amanda Hardy, Angela Kociolek, Anthony~P.
  Clevenger, Dan Smith, and Robert~J. Ament.
\newblock Wildlife-vehicle collision reduction study: Report to congress.
\newblock 2008.
\newblock URL
  \url{https://www.fhwa.dot.gov/publications/research/safety/08034/08034.pdf}.

\bibitem[Huijser et~al.(2009)Huijser, Duffield, Clevenger, Ament, and
  McGowen]{HuijserDCAM09}
Marcel~P. Huijser, John~W. Duffield, Anthony~P. Clevenger, Robert~J. Ament, and
  Pat~T. McGowen.
\newblock Cost–benefit analyses of mitigation measures aimed at reducing
  collisions with large ungulates in the united states and canada: a decision
  support tool.
\newblock \emph{Ecology and Society}, 14\penalty0 (2), 2009.
\newblock ISSN 17083087.
\newblock URL \url{http://www.jstor.org/stable/26268301}.

\bibitem[Jord{\'{a}}n and Schlotter(2015)]{JordanS15}
Tibor Jord{\'{a}}n and Ildik{\'{o}} Schlotter.
\newblock Parameterized complexity of spare capacity allocation and the
  multicost steiner subgraph problem.
\newblock \emph{J. Discrete Algorithms}, 30:\penalty0 29--44, 2015.
\newblock URL \url{https://doi.org/10.1016/j.jda.2014.11.005}.

\bibitem[Koana et~al.(2021)Koana, Korenwein, Nichterlein, Niedermeier, and
  Zschoche]{KKNNZ21}
Tomohiro Koana, Viatcheslav Korenwein, Andr\'{e} Nichterlein, Rolf Niedermeier,
  and Philipp Zschoche.
\newblock Data reduction for maximum matching on real-world graphs: Theory and
  experiments.
\newblock \emph{ACM J. Exp. Algorithmics}, 26:\penalty0 3.1--3.30, 2021.
\newblock URL \url{https://doi.org/10.1145/3439801}.

\bibitem[Kolmogorov(2009)]{Kolmogorov09}
Vladimir Kolmogorov.
\newblock Blossom {V:} a new implementation of a minimum cost perfect matching
  algorithm.
\newblock \emph{Math. Program. Comput.}, 1\penalty0 (1):\penalty0 43--67, 2009.
\newblock URL \url{https://doi.org/10.1007/s12532-009-0002-8}.

\bibitem[Lai et~al.(2011)Lai, Gomes, Schwartz, McKelvey, Calkin, and
  Montgomery]{LaiGSMCM11}
Katherine~J. Lai, Carla~P. Gomes, Michael~K. Schwartz, Kevin~S. McKelvey,
  David~E. Calkin, and Claire~A. Montgomery.
\newblock The steiner multigraph problem: Wildlife corridor design for multiple
  species.
\newblock In \emph{Proc.\ of 25th {AAAI}}. {AAAI} Press, 2011.
\newblock URL
  \url{http://www.aaai.org/ocs/index.php/AAAI/AAAI11/paper/view/3768}.

\bibitem[Loraamm and Downs(2016)]{LoraammD16}
Rebecca~W. Loraamm and Joni~A. Downs.
\newblock A wildlife movement approach to optimally locate wildlife crossing
  structures.
\newblock \emph{Int. J. Geogr. Inf. Sci.}, 30\penalty0 (1):\penalty0 74--88,
  2016.
\newblock URL \url{https://doi.org/10.1080/13658816.2015.1083995}.

\bibitem[Sawaya et~al.(2014)Sawaya, Kalinowski, and Clevenger]{SawayaKC14}
Michael~A. Sawaya, Steven~T. Kalinowski, and Anthony~P. Clevenger.
\newblock Genetic connectivity for two bear species at wildlife crossing
  structures in banff national park.
\newblock \emph{Proceedings of the Royal Society B: Biological Sciences},
  281\penalty0 (1780):\penalty0 20131705, 2014.
\newblock URL \url{https://www.ncbi.nlm.nih.gov/pmc/articles/PMC4027379/}.

\bibitem[Toussaint(1980)]{Toussaint80}
Godfried~T. Toussaint.
\newblock The relative neighbourhood graph of a finite planar set.
\newblock \emph{Pattern Recognit.}, 12\penalty0 (4):\penalty0 261--268, 1980.
\newblock URL \url{https://doi.org/10.1016/0031-3203(80)90066-7}.

\bibitem[Urban et~al.(2009)Urban, Minor, Treml, and Schick]{UrbanMTS09}
Dean~L. Urban, Emily~S. Minor, Eric~A. Treml, and Robert~S. Schick.
\newblock Graph models of habitat mosaics.
\newblock \emph{Ecology Letters}, 12\penalty0 (3):\penalty0 260--273, 2009.
\newblock URL \url{https://doi.org/10.1111/j.1461-0248.2008.01271.x}.

\bibitem[Woltz et~al.(2008)Woltz, Gibbs, and Ducey]{WoltzGD08}
Hara~W. Woltz, James~P. Gibbs, and Peter~K. Ducey.
\newblock Road crossing structures for amphibians and reptiles: Informing
  design through behavioral analysis.
\newblock \emph{Biological Conservation}, 141\penalty0 (11):\penalty0
  2745--2750, 2008.
\newblock ISSN 0006-3207.
\newblock URL \url{https://doi.org/10.1016/j.biocon.2008.08.010}.

\end{thebibliography}
\endgroup}

\appendix
\renewcommand\thefigure{\thesection.\arabic{figure}}
\renewcommand\thetable{\thesection.\arabic{table}}
\section*{\Large Appendix}
\section{Additional figures and tables}
We provide additional figures and tables for a more in-depth overview over our results.
In~\cref{fig:vsilp2},
we show how~$\Amcm$ and~$\Amchm$ perform (in terms of time) against~$\Agen$.
Interestingly,
is some few cases for real-world instances with 200 cycle habitats,
$\Agen$ runs faster than~$\Amchm$.
In~\cref{fig:approxvsopt},
we show how~$\Aapx$ performs (in terms of solution quality) against OPT.
One can see that there are no big fluctuation from the diagonal (approximation ratio 1).
Moreover,
larger solution costs and large set of habitats increases the approximation ratio.

\begin{figure*}[t]
    \includegraphics[width=\textwidth]{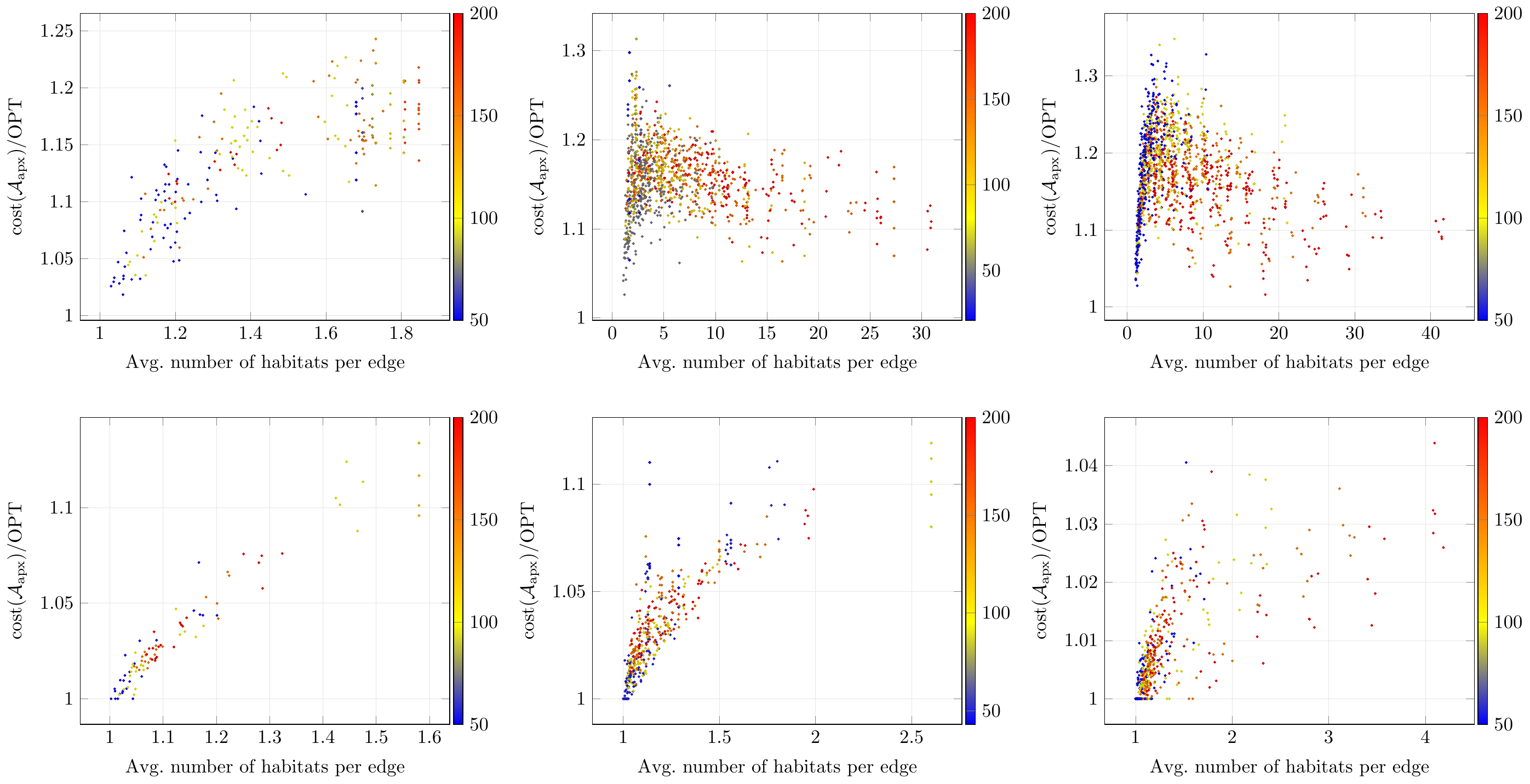}
    \caption{
    Approximation algorithm ratio against our intersection measure.
    (Top) Real-world. 
    (Bottom) Artificial. 
    (Left) Faces. 
    (Middle) Cycles. 
    (Right) Random Walks.
    }
    \label{fig:apx-cost-vs-intersect}
\end{figure*}

\begin{figure*}[t]
    \includegraphics[width=\textwidth]{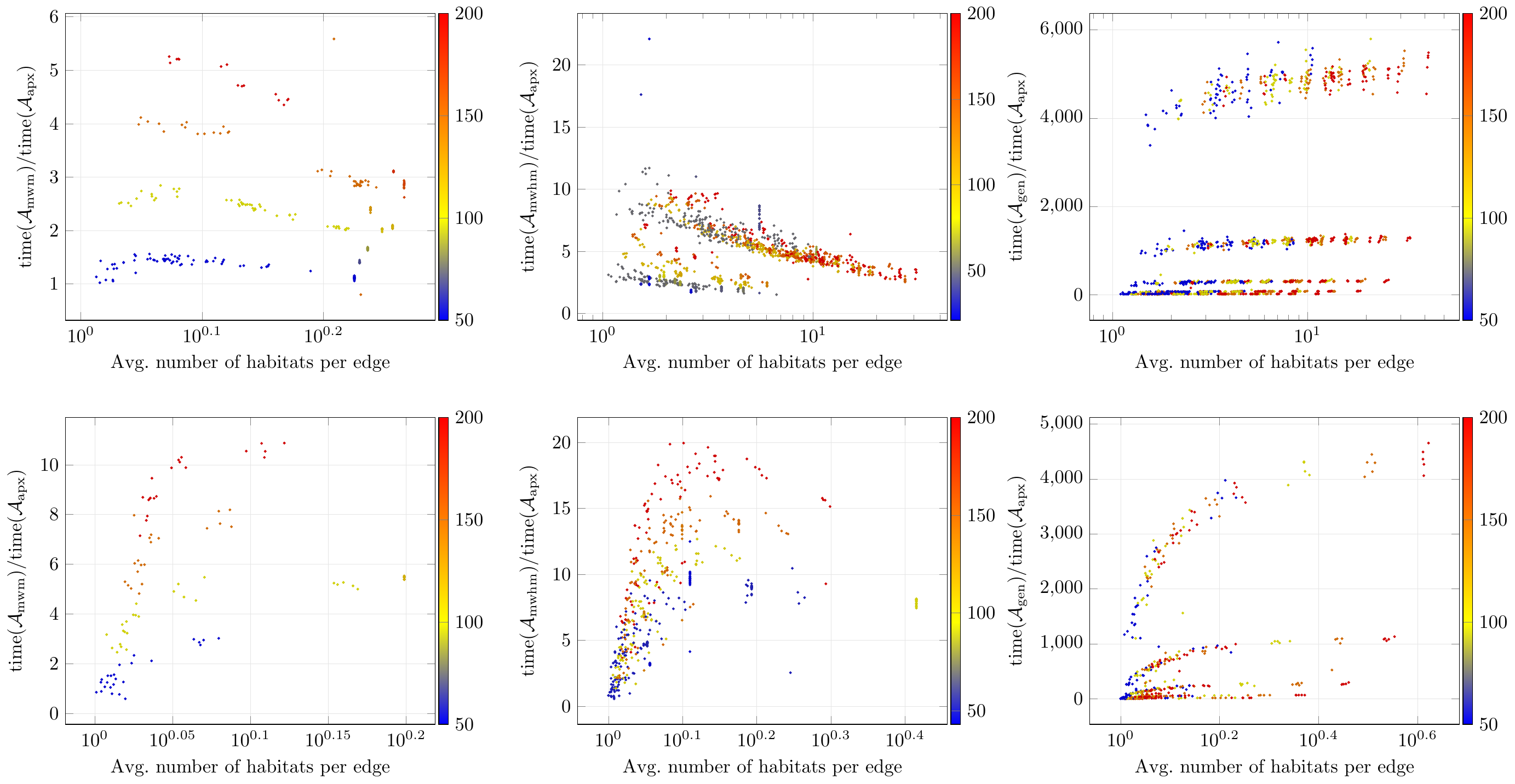}
    \caption{
    Quotient of running times of best optimal algorithm and approximation algorithm
    against our intersection measure.
    (Top) Real-world. 
    (Bottom) Artificial. 
    (Left) Faces. 
    (Middle) Cycles. 
    (Right) Random Walks.
    }
    \label{fig:apx-time-vs-intersect}
\end{figure*}

\begin{table*}[t]\centering
 \caption{Summary of our results regarding Faces on NW. 
 ``btime'' is short for ``building time''.}
 \begin{tabular}{@{}r||r|rrr|rrrrr@{}}\toprule
	 & \small intersect & \small OPT & \small cost($\Aapx$) & \small ratio & \small time($\Aapx$) & \small time($\Amcm$) & \small time($\Amchm$) & \small time($\Agen$) & \small btime($\Agen$)\\\midrule\midrule
 50 & 1.048 & 401.8 & 415.4 & 1.034 & 0.004 & 0.004 & 0.008 & 0.022 & 0.016\\
 100  & 1.092 & 762 & 800.4 & 1.05 & 0.006 & 0.014 & 0.023 & 0.037 & 0.031\\
 150  & 1.14 & 1103.6 & 1196.2 & 1.084 & 0.007 & 0.03 & 0.046 & 0.052 & 0.044\\
 200  & 1.195 & 1381.8 & 1536.8 & 1.112 & 0.009 & 0.049 & 0.071 & 0.063 & 0.054\\
  \bottomrule
 \end{tabular}
 \label{tab:summary_nw}
\end{table*}

\begin{figure*}[t]
    \includegraphics[width=1\textwidth]{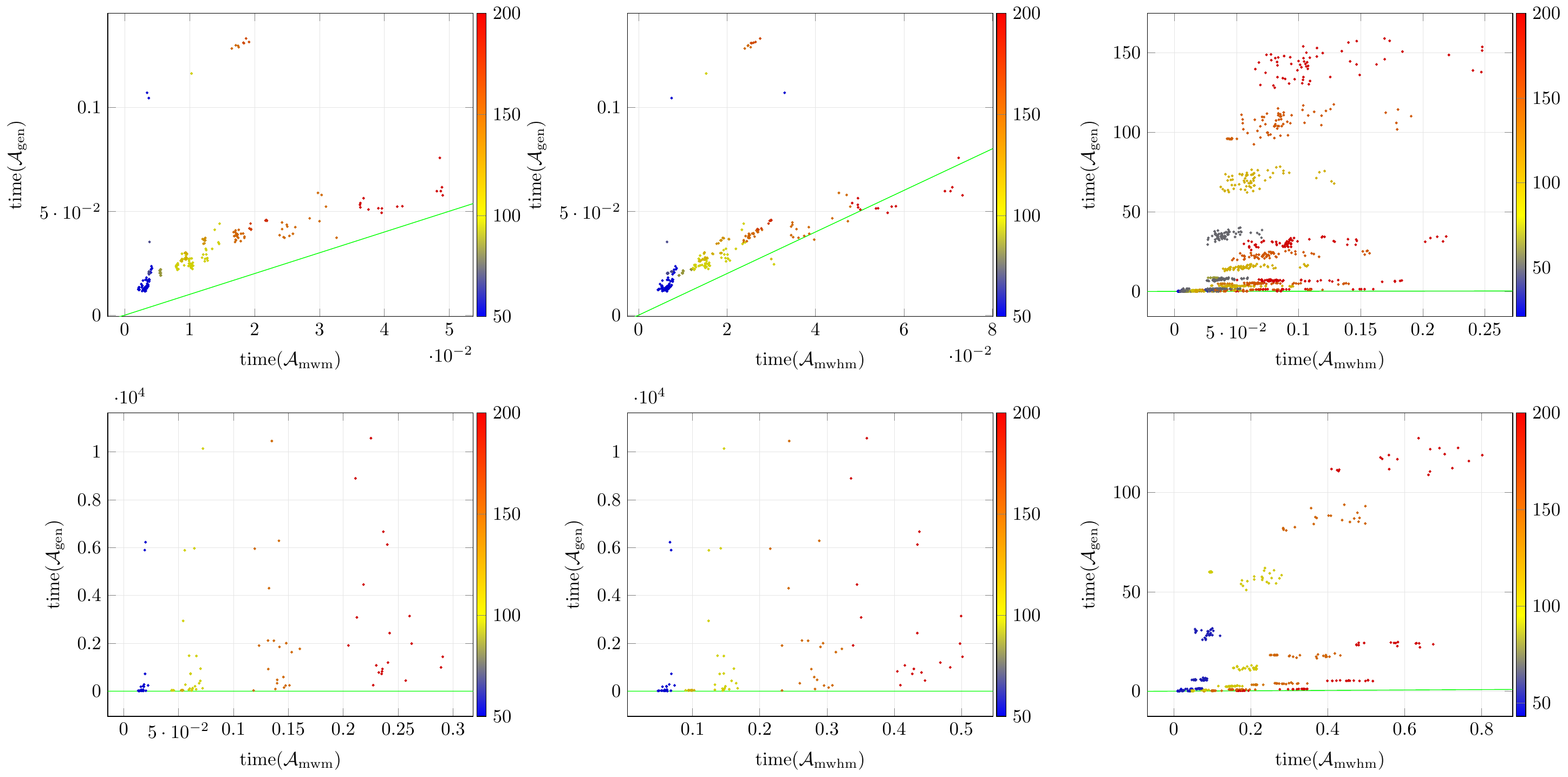}
    \caption{Comparison of the running times of $\Amcm$ and~$\Amchm$ versus~$\Agen$.
    (Top) Real-world. 
    (Bottom) Artificial. 
    (Left and middle) Faces instances. 
    (Right) Cycles instances.}
    \label{fig:vsilp2}
\end{figure*}

\begin{figure*}[t]
    \includegraphics[width=1\textwidth]{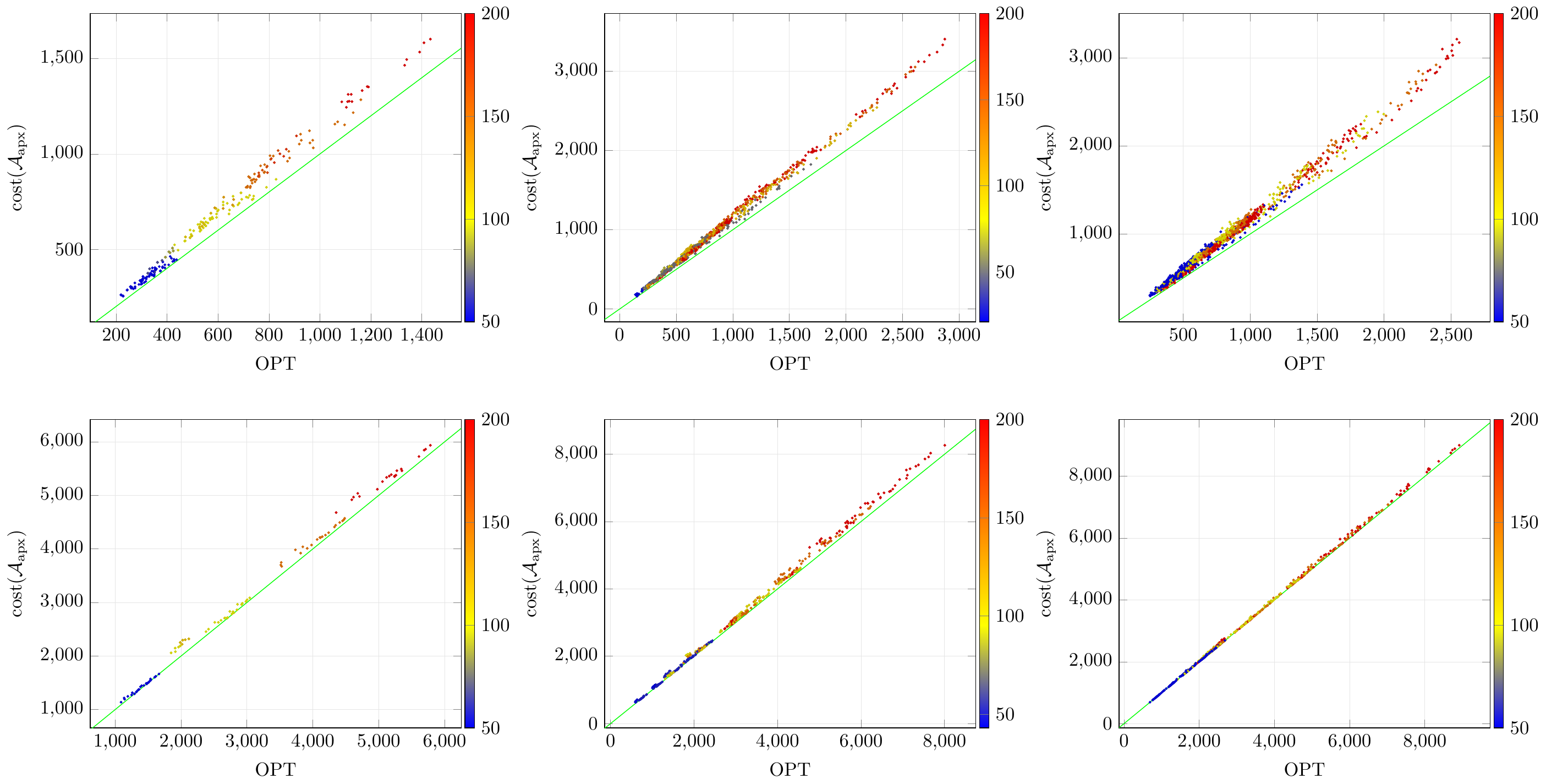}
    \caption{Approximation solution against OPT.
    (Top) Real-world. 
    (Bottom) Artificial. 
    (Left) Faces. 
    (Middle) Cycles. 
    (Right) Random Walks.}
    \label{fig:approxvsopt}
\end{figure*}

\ifshort{}
\appendixProofText
\fi{}

\end{document}